\documentclass[11 pt]{elsarticle}
\usepackage{graphicx, url}
\usepackage{latexsym}
\usepackage{amsfonts}
\usepackage{amsmath}
\usepackage{enumerate}
\usepackage{fullpage}
\usepackage{subfigure}
\usepackage{color}
\usepackage{array} 
\usepackage{arydshln}

\usepackage{amsthm}
\newtheorem{theorem}{Theorem}[section]
\newtheorem{lemma}[theorem]{Lemma}
\newtheorem{coro}[theorem]{Corollary}

\newtheorem{obs}[theorem]{Observation}

\newcommand{\defi}[1]{\emph{#1}}
\newcommand{\comp}[1]{\overline{#1}}
\newcommand{\Vusk}{Vu\v{s}kovi\'c} 
\newcommand{\EH}{Erd\H{o}s-Hajnal}
\newcommand{\set}[1]{\{#1\}}
\renewcommand{\choose}[2]{\left( \genfrac{}{}{0pt}{}{#2}{#1} \right) }

\begin{document}
\title{Clique-Stable Set Separation in Perfect Graphs with no Balanced Skew-Partitions\tnoteref{stint}}

\tnotetext[stint]{This work is partially supported by ANR project \textsc{Stint} under reference ANR-13-BS02-0007.}

\author[ens]{Aur\'elie Lagoutte\corref{cor1}}
\ead{aurelie.lagoutte@ens-lyon.fr}

\author[ens]{Th\'eophile Trunck}
\ead{theophile.trunck@ens-lyon.fr}

\cortext[cor1]{Corresponding author}

\address[ens]{LIP, UMR 5668 ENS Lyon - CNRS - UCBL - INRIA, Universit\'e de Lyon, 46, all\'ee de l'Italie, 69364 Lyon France.}

\date{\today}

\begin{frontmatter}

\begin{abstract}

  Inspired by a question of Yannakakis on the Vertex Packing polytope of perfect graphs, we  study the Clique-Stable Set separation in a non-hereditary subclass of perfect graphs.
A cut $(B,W)$ of $G$ (a bipartition  of  $V(G)$) \emph{separates} a clique $K$ and a stable set $S$ if $K\subseteq B$ and $S\subseteq W$. A \emph{Clique-Stable Set separator} is a family of cuts such that for every clique $K$, and for every stable set $S$ disjoint from $K$, there exists a cut in the family that separates $K$ and $S$. Given a class of graphs, the question is to know whether every graph of the class admits a Clique-Stable Set separator containing only polynomially many cuts. It was recently proved to be false for the class of all graphs (G\"o\"os 2015), but it remains open for perfect graphs, which was Yannakakis' original question.
  Here we investigate this problem on  perfect graphs with no balanced skew-partition; the \emph{balanced skew-partition} was introduced in the decomposition theorem of Berge graphs which led to the celebrated proof of the Strong Perfect Graph Theorem. Recently, Chudnovsky, Trotignon, Trunck and Vu\v{s}kovi\'c proved that forbidding this unfriendly decomposition permits to recursively decompose Berge graphs (more precisely, Berge trigraphs) using 2-join and  complement 2-join until reaching a ``basic'' graph, and in this way, they found an efficient  combinatorial  algorithm to color those graphs.
  
  We apply their decomposition result to prove that perfect graphs with no balanced skew-partition admit a quadratic-size Clique-Stable Set separator, by taking advantage of the good behavior of 2-join with respect to this property.
   We then generalize this result and prove that 
   the Strong \EH \ property holds in this class, which means that every such graph  has a linear-size biclique or complement biclique. This  is remarkable since the property does not hold for all perfect graphs (Fox 2006), and this is motivated here by the following statement:
when the Strong \EH \ property holds in a hereditary class of graphs, then both the \EH \ property and the polynomial Clique-Stable Set separation hold.
 Finally,
  we define the generalized $k$-join and generalize both our results
  on classes of graphs admitting such a decomposition. 
\end{abstract}

\begin{keyword}
Clique-Stable Set separation \sep perfect graph \sep trigraph \sep 2-join 
\end{keyword}
\end{frontmatter}

\section{Introduction}

In 1991, Yannakakis \cite{Yannakakis91} studied the
Vertex Packing polytope of a graph (also called the Stable Set polytope),
and asked for the existence of an extended formulation, that is to say
a simpler polytope in higher dimension whose projection would be the
Vertex Packing polytope. He then focused on perfect
graphs, for which the non-negativity and the clique constraints
suffice to describe the Vertex Packing polytope. This led him to a
communication complexity problem which can be restated as follows: does
there exist a  family $F$ of polynomially many cuts (a \emph{cut} is a bipartition of the vertices of the graph)  such that, for every
clique $K$ and every stable set $S$ of the graph that do not
intersect, there exists a cut $(B,W)$ of $F$ that \emph{separates} $K$
and $S$, meaning $K\subseteq B$ and $S \subseteq W$? Such a family of
cuts separating all the cliques and the stable sets is called a
Clique-Stable Set separator (CS-separator for short). The existence of
a polynomial CS-separator (called the \emph{Clique-Stable Set separation}, or CS-separation)
is a necessary condition for the existence of an extended
formulation. Yannakakis showed that both exist for several subclasses
of perfect graphs, such as comparability graphs and their complements,
chordal graphs and their complements, and Lov\'asz proved it for
a generalization of series-parallel graphs called $t$-perfect
graphs \cite{Lovasz94}. However, the problem remains open for perfect graphs in general.

Twenty years have passed since Yannakakis introduced the problem and
several results have shed some light on the problem. First of all, a negative
result due to Fiorini et al. \cite{Fiorini11} asserts that there does not exist
an extended formulation for the Vertex Packing polytope for all
graphs. 
Furthermore on the negative side, G\"o\"os recently proved the existence of graphs for which no polynomial CS-separator exists \cite{Goos}.
This pushes us further to the study of perfect graphs, for
which great progress has been made. The most famous one is the Strong
Perfect Graph Theorem \cite{SPGT}, proving that a graph is perfect if and only if
it is Berge, that is to say it contains no odd hole and no odd antihole (as induced subgraph).
It was proved by Chudnovsky, Robertson,
Seymour and Thomas, and their proof relies on a decomposition
theorem \cite{SPGT, trigraphs}, whose statement can be summed up as follows: every Berge graph is either in some basic class, or has some
kind of decomposition (\emph{2-join}, \emph{complement 2-join} or \emph{balanced skew-partition}).
It seems natural to take advantage of this decomposition theorem to try to solve Yannakakis' question on perfect graphs. We will see that the 2-join and its complement behave well with respect to the Clique-Stable Set separation, whereas the balanced skew-partition does not.

Consequently, instead of proving the CS-separation for all perfect graphs, we would like to reach a weaker goal and prove the CS-separation for perfect graphs that can be recursively decomposed using 2-joins or complement 2-joins until reaching a basic class. 
Because of the decomposition theorem, a natural candidate is the class of Berge graphs with no balanced skew-partition, which has already been studied in 
\cite{bergefreebsp}, where
Chudnovsky, Trotignon, Trunck and \Vusk \  aimed at finding a  combinatorial polynomial-time algorithm to color perfect graphs. 
They proved that if a Berge graph is not basic and has no balanced skew-partition, then its decomposition along a 2-join gives two Berge graphs which still have no balanced skew-partition\footnote{In fact, the correct statement must be stated in terms of trigraphs instead of graphs.}.
This, together with a deeper investigation,  led them to  a combinatorial polynomial-time algorithm to compute the Maximum Weighted Stable Set in Berge graphs with no balanced skew-partition, from which they deduced a coloring algorithm.

They used a powerful concept,  called \emph{trigraph}, which is a generalization of a graph. It was introduced by  Chudnovsky in her PhD thesis \cite{TheseMaria, trigraphs}  to simplify the statement and the proof of the Strong Perfect Graph Theorem. Indeed, the original statement of the decomposition theorem provided five different outcomes, but she proved that one of them (the homogeneous pair) is not necessary. Trigraphs are also very useful in the study of bull-free graphs \cite{BullFreeI, BullFreeII-III, FPTBullFree} and claw-free graphs \cite{ClawFree}.
Using the previous study of Berge trigraphs with no balanced skew-partition from \cite{bergefreebsp}, we prove that Berge graphs with no balanced skew-partition have a polynomial CS-separator. We then observe that we can obtain the same result by relaxing 2-join to a more general kind of decomposition, which we call  \emph{generalized $k$-join}.

Besides, the Clique-Stable Set separation has been recently studied in \cite{Bousquet13}, where the authors exhibit polynomial CS-separators for several classes of graphs, namely random graphs, $H$-free graphs where $H$ is a split graph, $P_5$-free graphs, and $(P_k, \comp{P_k})$-free graphs (where $P_k$ denotes the path on $k$ vertices and $\comp{P_k}$ its complement). This last result was obtained as a consequence of \cite{Thomasse13} where the same authors prove
that the Strong \EH \ property holds in this class, which implies the Clique-Stable Set separation and the \EH \ property 
(provided that the class is  closed under taking induced subgraphs \cite{Alon05, FoxPach08}).
The \EH \ conjecture asserts that for every graph $H$, there exists $\varepsilon>0$ such that every $H$-free graph $G$ admits a clique or a stable set of size $|V(G)|^\varepsilon$. Several attempts have been made to prove this conjecture (see \cite{EHSurvey} for a survey). In particular, Fox and Pach introduced to this end the Strong \EH \ property \cite{FoxPach08}: a \emph{biclique} is a pair  of disjoint subsets of vertices $V_1, V_2$ such that $V_1$ is complete to $V_2$; the \emph{Strong \EH \ property} holds in a class $\mathcal{C}$ if there exists a constant $c>0$ such that for every $G\in \mathcal{C}$, $G$ or $\comp{G}$ admits a biclique $(V_1, V_2)$ with $|V_1|, |V_2|\geq c\cdot |V(G)|$.
In other words, Fox and Pach ask for a linear-size biclique in $G$ or in $\comp{G}$, 
instead of  a polynomial-size clique in $G$ or in $\comp{G}$, as in the definition of the \EH \ property. 
 Even though the \EH \ property is trivially true for perfect graphs with $\varepsilon=1/2$ (since $|V(G)|\leq \alpha(G)\chi(G)$ and $\chi(G)=\omega(G)$), Fox proved that a subclass of comparability graphs (and thus, of perfect graphs) does not have the Strong \EH \ property \cite{Fox06}. 
Consequently, it is worth investigating this property in the subclass of perfect graphs under study. We prove that perfect graphs with no balanced skew-partition have the Strong  \EH \ property.
 Moreover we combine both generalizations and prove that trigraphs that can be recursively decomposed with generalized $k$-join also have the Strong \EH \ property.  It should be noticed that the class of Berge graphs with no balanced skew-partition is not hereditary (\emph{i.e.} not closed under taking induced subgraphs) because removing a vertex may create a balanced skew-partition, so the CS-separation is not a consequence of the Strong \EH \ property and needs a full proof. 

The fact that the Strong \EH \ property holds in Berge graphs with no balanced skew-partition shows that this subclass is much less general than the whole class of perfect graphs. 
This observation is confirmed by another recent work by Penev \cite{2CliqueColoSansSkew} who also studied the class of Berge graphs with no balanced skew-partition and proved that they admit a 2-clique-coloring (\emph{i.e.} there exists a non-proper coloration with two colors such that every inclusion-wise maximal clique is not monochromatic). 
Perfect graphs in general are not 2-clique-colorable, but they were conjectured to be 3-clique-colorable; Charbit et al. recently disproved it by constructing perfect graphs with arbitrarily high clique-chromatic number \cite{CliqueChromaticPerfect}.

Let us now define what is a balanced skew-partition in a graph and then compare the class of perfect graphs with no balanced skew-partition to classical hereditary subclasses of perfect graphs.
A graph $G$ has a \emph{skew-partition} if $V(G)$ can be partitioned into $(A,B)$ such that neither $G[A]$ nor $\comp{G[B]}$ is connected. Moreover, the \emph{balanced} condition, although essential in the proof of the Strong Perfect Graph Theorem, is rather technical: the partition is \emph{balanced} if every path in $G$ of length at least 3, with ends in $B$ and interior in $A$, and every path in $\comp{G}$, with ends in $A$ and interior in $B$, has even length.
Observe now  for instance that $P_4$, which is a bipartite,  chordal and comparability graph, has a balanced skew-partition (take the extremities as the non-connected part $A$, and the two middle vertices as the non-anticonnected part $B$). However, $P_4$ is an induced subgraph of $C_6$, which has no skew-partition. So sometimes one can \emph{kill} all the balanced skew-partitions by adding some vertices. Trotignon and Maffray proved that given a basic graph $G$ on $n$ vertices having a balanced skew-partition, there exists a basic graph $G'$ on $\mathcal{O}(n^2)$ vertices which has no balanced skew-partition and contains $G$ as an induced subgraph \cite{TrotignonMaffray}.
Some degenerated cases are to be considered: graphs with at most 3 vertices as well as cliques and stable sets do not have a balanced skew-partition.
Moreover,
Trotignon showed \cite{Trotignon08} that every double-split graph does not have a balanced skew-partition. In addition to this, observe that any clique-cutset of size at least 2 gives rise to a balanced skew-partition: as a consequence, paths, chordal graphs and cographs always have a balanced skew-partition, up to a few degenerated cases. Table \ref{tab: BSP} compares the class of Berge graphs with no balanced skew-partition with some examples of well-known subclasses of perfect graphs. In particular, there exist two non-trivial perfect graphs lying in none of the above mentioned classes (basic graphs, chordal graphs, comparability graphs, cographs), one of them having a balanced skew-partition, the other not having any.

\begin{table}[t]

\center
{\renewcommand{\arraystretch}{1.1}\begin{tabular}{|m{5cm}|c|  c | }
\hline
 & With a BSP & With no BSP \\ \hline \hline
 Bipartite graph & $P_4$ & $C_4$ \\ \hline
  Compl. of a bipartite graph & $P_4$ & $C_4$ \\ \hline
   Line graph of a bip. graph & $P_4$ & $C_4$ \\ \hline
    Complement of a  line graph of a bip. graph  & $P_4$ & $C_4$ \\ \hline
     Double-split  & None & $C_4$ \\ \hline
      Comparability graph & $P_4$ & $C_4$ \\ \hline
     Path & $P_k$ for $k\geq 4$ & None \\ \hline
     Chordal  & All (except deg. cases) & $K_t, S_t, \comp{C_4}$ , $t\geq 4$ \\ \hline
        Cograph & All (except deg. cases) & $K_t, S_t, C_4, \comp{C_4}$, $t\geq 4$ \\ \hline
        None of the classes above \phantom{xxxxxx}  & \parbox{3.5cm}{Worst Berge Graph Known so Far} & Zambelli's graph \\ \hline
\end{tabular}}

\caption{Classical subclasses of perfect graphs compared with perfect graphs with no balanced skew-partition (BSP for short). Graphs with less than 4 vertices are not considered. See Figure \ref{fig: zambelli, WBGKSF} for a description of the Worst Berge Graph Known So Far and Zambelli's graph.}
\label{tab: BSP}
\end{table}

\begin{figure}
\center

\subfigure[The Worst Berge Graph Known So Far. (discovered by Chudnovsky and Seymour; displayed in \cite{TrotSurvey}). Red edges (resp. blue edges, green edges) go to red (resp. blue, green) vertices.]{\hspace{40pt}\includegraphics[scale=0.6]{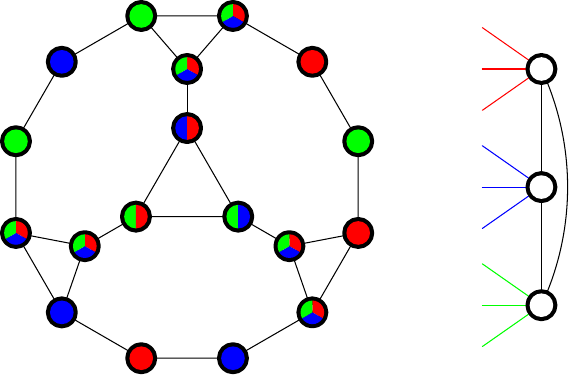} \hspace{40pt}} 
\hspace{20pt}
\subfigure[Zambelli's graph.]{\includegraphics[scale=1.25]{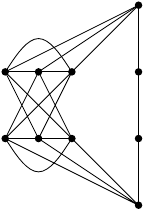}}

\caption{The two non-trivial perfect graphs dealt with in Table \ref{tab: BSP}: the first has a BSP, the second one does not.}
\label{fig: zambelli, WBGKSF}
\end{figure}

We start in Section \ref{sec: def} by introducing trigraphs and all related definitions. In Section \ref{sec: decomposition}, we state the decomposition theorem from \cite{bergefreebsp} for Berge trigraphs with no balanced skew-partition.
The results come in the last two sections: Section \ref{sec: CS-sep} is concerned with finding polynomial-size \mbox{CS-separators} in Berge trigraphs with no balanced skew-partition, and then with extending this result to  classes of trigraphs closed by generalized $k$-join, provided that the \emph{basic} class admits polynomial-size CS-separators. As for Section~\ref{sec: SEH}, it is dedicated to proving that the Strong \EH \ property holds in perfect graphs with no balanced skew-partition, and then in classes of trigraphs closed by generalized $k$-join (with a similar assumption on the \emph{basic} class).

%

\section{Definitions}
\label{sec: def}

We first need to introduce trigraphs: this is a generalization of graphs where a new kind of adjacency between vertices is defined: the \emph{semi-adjacency}. The intuitive meaning of a pair of semi-adjacent vertices, also called a \emph{switchable pair}, is that in some situations, the vertices are considered as adjacent, and in some other situations, they are considered as non-adjacent. This implies to be very careful about terminology, for example in a trigraph two vertices are said \emph{adjacent} if there is a ``real" edge between them but also if they are semi-adjacent. What if we want to speak about ``really adjacent" vertices, in the old-fashioned way? The dedicated terminology is \emph{strongly adjacent}, adapted to \emph{strong neighborhood}, \emph{strong clique} and so on. 

Because of this, we need to redefine all the usual notions on graphs to adapt them on trigraphs, which we do in the the next subsection. For example, a trigraph is not Berge if we can turn each switchable pair into a strong edge or a strong antiedge in such a way that the resulting graph has an odd hole or an odd antihole.
%
%
 Moreover, the trigraphs we are interested in come from decomposing Berge graphs along 2-joins. As we will see in the next section, this leads to the appearance of only  few switchable pairs, or at least distant switchable pairs. This property is useful both for decomposing trigraphs and for proving the CS-separation in basic classes, so we work in the following on a restricted class of Berge trigraphs, which we denote $\mathcal{F}$.
 In a nutshell\footnote{The exact definition is in fact much more precise.}, it is the class of  Berge trigraphs whose \emph{switchable components} (connected components of the graph obtained by keeping only switchable pairs) are paths of length at most 2.

Let us now give formal definitions.

\subsection{Trigraphs}

For a set $X$, we denote by $\choose{2}{X}$ the set of all subsets of
$X$ of size~2. For brevity of notation an element $\{ u,v \}$ of $\choose{2}{X}$ is also denoted by $uv$ or $vu$. A \defi{trigraph} $T$
consists of a finite set $V(T)$, called the \defi{vertex set} of $T$,
and a map $\theta : {\choose{2}{V(T)}} \longrightarrow \{ -1,0,1 \}$,
called the \defi{adjacency function}.

Two distinct vertices of $T$ are said to be \defi{strongly adjacent} if
$\theta(uv)=1$, \defi{strongly antiadjacent} if $\theta(uv)=-1$, and
\defi{semiadjacent} if $\theta(uv)=0$. We say that $u$ and $v$ are
\defi{adjacent} if they are either strongly adjacent, or semiadjacent;
and \defi{antiadjacent} if they are either strongly antiadjacent, or
semiadjacent. An \defi{edge} (\defi{antiedge}) is a pair of adjacent
(antiadjacent) vertices. If $u$ and $v$ are adjacent (antiadjacent),
we also say that $u$ is \defi{adjacent (antiadjacent) to} $v$, or that
$u$ is a \defi{neighbor (antineighbor)} of $v$. The \defi{open neighborhood} $N(u)$ of $u$ is the set of neighbors of $u$, and the \defi{closed neighborhood} $N[u]$  of $u$ is $N(u)\cup \{u\}$.
If $u$ and
$v$ are strongly adjacent (strongly antiadjacent), then $u$ is a \defi{
  strong neighbor (strong antineighbor)} of $v$. Let $\sigma(T)$ the set of all
semiadjacent pairs of $T$. Thus, a trigraph $T$ is a graph if
$\sigma(T)$ is empty. A pair $\{u, v\} \subseteq V(T)$ of distinct
vertices is a \defi{switchable pair} if $\theta(uv) = 0$, a
\defi{strong edge} if $\theta(uv) = 1$ and a \defi{strong antiedge} if
$\theta(uv) = -1$.  An edge $uv$ (antiedge, strong edge, strong
antiedge, switchable pair) is \defi{between} two sets $A \subseteq
V(T)$ and $B \subseteq V(T)$ if $u\in A$ and $v \in B$ or if $u \in B$
and $v \in A$.

Let $T$ be a trigraph. The \defi{complement} $\overline{T}$ of $T$ is a
trigraph with the same vertex set as $T$, and adjacency function
$\overline{\theta}=-\theta$. Let $A \subset V(T)$ and $b
\in V(T) \setminus A$. We say that $b$ is \defi{strongly complete} to
$A$ if $b$ is strongly adjacent to every vertex of $A$; $b$ is \defi{
 strongly anticomplete} to $A$ if $b$ is strongly antiadjacent to
every vertex of $A$; $b$ is \defi{complete} to $A$ if $b$ is adjacent
to every vertex of $A$; and $b$ is \defi{anticomplete} to $A$ if $b$ is
antiadjacent to every vertex of $A$. For two disjoint subsets $A,B$
of $V(T)$, $B$ is \defi{strongly complete (strongly anticomplete,
 complete, anticomplete)} to $A$ if every vertex of $B$ is strongly
complete (strongly anticomplete, complete, anticomplete) to $A$.

A \defi{clique} in $T$ is a set of vertices all pairwise adjacent, and
a \defi{strong clique} is a set of vertices all pairwise strongly
adjacent. A \defi{stable set} is a set of vertices all pairwise
antiadjacent, and a \defi{strong stable set} is a set of vertices all
pairwise strongly antiadjacent. For $X \subset V(T)$ the trigraph
\defi{induced by $T$ on $X$} (denoted by $T[X]$) has vertex set $X$,
and adjacency function that is the restriction of $\theta$ to $\choose{2}{X}$. Isomorphism between trigraphs is defined in the natural
way, and for two trigraphs $T$ and $H$ we say that $H$ is an \defi{
 induced subtrigraph} of $T$ (or $T$ \defi{contains $H$ as an induced
 subtrigraph}) if $H$ is isomorphic to $T[X]$ for some $X \subseteq
V(T)$. Since in this paper we are only concerned with the induced subtrigraph
containment relation, we say that \defi{$T$ contains~$H$} if $T$
contains $H$ as an induced subtrigraph. We denote by $T\setminus X$
the trigraph $T[V(T) \setminus X]$.

Let $T$ be a trigraph. A \defi{path} $P$ of $T$ is a sequence of
distinct vertices $p_1, \dots, p_k$ such that either $k=1$, or for $i,
j \in \{1, \ldots, k\}$, $p_i$ is adjacent to $p_j$ if $|i-j|=1$ and
$p_i$ is antiadjacent to $p_j$ if $|i-j|>1$. We say that $P$ is a
path \defi{from $p_1$ to $p_k$}, its \defi{interior} is the set
$\{p_2, \ldots, p_{k-1}\}$, and the \defi{length} of $P$ is
$k-1$.  Observe that, since a graph is also a
trigraph, it follows that a path in a graph, the way we have defined
it, is what is sometimes in literature called a chordless path.

A \defi{hole} in a trigraph $T$ is an induced subtrigraph $H$ of $T$
with vertices $h_1, \ldots, h_k $ such that $k \geq 4$, and for $i,j
\in \{1, \ldots, k\}$, $h_i$ is adjacent to $h_j$ if $|i-j|=1$ or
$|i-j|=k-1$; and $h_i$ is antiadjacent to $h_j$ if $1<|i-j|<k-1$. The
{\em length} of a hole is the number of vertices in it. An \defi{antipath} (\defi{antihole})
in $T$ is an induced subtrigraph of $T$ whose complement is a path
(hole) in $\overline{T}$.

A \defi{semirealization} of a trigraph $T$ is any trigraph $T'$ with
vertex set $V(T)$ that satisfies the following: for all $uv \in
{\choose{2}{V(T)}}$, if $uv$ is a strong edge in $T$, then it is also a strong edge in $T'$, and if $uv$ is a strong antiedge in $T$, then  it is also a strong antiedge in $T'$.
  Sometimes we will describe
a semirealization of $T$ as an \defi{assignment of values} to
switchable pairs of $T$, with three possible values: ``strong edge'',
``strong antiedge'' and ``switchable pair''.  A \defi{realization} of
$T$ is any graph that is semirealization of $T$ (so, any
semirealization where all switchable pairs are assigned the value
``strong edge'' or ``strong antiedge'').  The realization where all switchable pairs are assigned the value ``strong edge''  is called the \defi{full realization} of~$T$.

Let $T$ be a trigraph. For $X \subseteq V(T)$, we say that $X$ and
$T[X]$ are \defi{connected} (\defi{anticonnected}) if the 
full realization of $T[X]$ ($\comp{T[X]}$) is
connected. A \defi{connected component} (or simply \defi{component}) of
$X$ is a maximal connected subset of $X$, and an \defi{anticonnected
component} (or simply \defi{anticomponent}) of $X$ is a maximal
anticonnected subset of $X$.

A trigraph $T$ is  \defi{Berge} if it contains no odd hole and no odd
antihole. Therefore, a trigraph is Berge if and only if its complement
is. We observe that $T$ is Berge if and only if every realization
(semirealization) of $T$ is Berge.

Finally let us define the class of trigraphs we are working on.
Let $T$ be a trigraph, denote by $\Sigma(T)$ the graph with vertex set
$V(T)$ and edge set $\sigma(T)$ (the switchable pairs of $T$).  The
connected components of $\Sigma(T)$ are called the \defi{switchable
  components} of $T$.  Let \defi{$\mathcal{F}$} be the class of Berge
trigraphs such that the following hold:
\begin{itemize}
\item Every switchable component of $T$ has at most two edges (and
therefore no vertex has more than two neighbors in $\Sigma(T)$).
\item Let $v \in V(T)$ have degree two in $\Sigma(T)$, denote its
neighbors by $x$ and~$y$. Then either $v$ is strongly complete to
$V(T) \setminus \{v, x, y\}$ in $T$, and $x$ is strongly adjacent to
$y$ in $T$, or $v$ is
strongly anticomplete to $V(T) \setminus \{v, x, y\}$ in $T$, and $x$
is strongly antiadjacent to $y$ in $T$.
\end{itemize} 

Observe that $T\in \mathcal{F}$ if and only if $\overline{T}\in
\mathcal{F}$.

\subsection{Clique-Stable Set separation}

Let $T$ be a trigraph. 
A \defi{cut} is a partition of $V(T)$ into two parts $B,W\subseteq V(T)$ (hence $W=V(T)\setminus B)$. It \defi{separates} a clique $K$ and a stable set $S$
if $K \subseteq B$ and $S \subseteq W$. 
Sometimes we will call $B$ the \defi{clique side} of the cut and $W$ the \defi{stable set side} of the cut. 
In order to have a stronger assumption when applying induction hypothesis later on in the proofs, we choose to separate not only strong cliques and strong stable sets, but all cliques and all stable sets:
we say that a family $F$ 
of cuts is a \defi{CS-separator} if for every (not necessarily strong) clique $K$ and every (not necessarily strong) stable set $S$ which do not intersect, there exists a cut in $F$ that separates $K$ and $S$.
Finding a CS-separator is a self-complementary problem: suppose that there exists a CS-separator of size $k$ in $T$, then we build a CS-separator of size $k$ in $\overline{T}$ by turning  every cut $(B,W)$ into the cut $(W,B)$.

In a graph, a clique and a stable set can intersect on at most one vertex. This property is useful to prove that we only need to focus on inclusion-wise maximal cliques and inclusion-wise maximal stable sets (see \cite{Bousquet13}). This is no longer the case for trigraphs, for which a clique and a stable set can intersect on a switchable component $V'$, provided this component contains only switchable pairs, (\emph{i.e.}  for every $u,v \in V'$, $u=v$ or $uv\in \sigma(T)$). 
However, when restricted to trigraphs of $\mathcal{F}$, a clique and a stable set can intersect on at most one vertex or one switchable pair, so we can still derive a similar result:

\begin{obs}\label{clique max trigraph}
If a trigraph $T$ of $\mathcal{F}$ admits a family $F$ of cuts separating all the  inclusion-wise maximal cliques and the inclusion-wise maximal stable sets, then it admits a CS-separator of size at most $|F|+O(n^2)$.
\end{obs}

\begin{proof}
Start with $F'=F$ and add the following cuts to $F'$: for every $x\in V(T)$, add the cut $(N[x], V(T)\setminus N[x])$ and  the cut $(N(x),V(T)\setminus  N(x))$. For every switchable pair $xy$,
add the four cuts of type $(U_i, V(T)\setminus U_i)$ with
 
$$ U_1=N[x]\cap N[y], \ 
 U_2=N[x]\cap N(y), \
U_3=N(x)\cap N[y], \
 U_4=N(x)\cap N(y)\ . $$

Let $K$ be a clique and $S$ be a stable set disjoint from $K$, and let $K'$ (resp. $S'$) be an inclusion-wise maximal clique (resp. stable set) containing $K$ (resp. $S$).
 Three cases are to be considered. First, assume that $K'$ and $S'$ do not intersect, then there is a cut in $F$ that separates $K'$ from $S'$ (thus $K$ from $S$). Second, assume that $K'$ and $S'$ intersect on a vertex $x$ : if $x \in K$, then $K\subseteq N[x]$ and $S\subseteq V(T)\setminus N[x]$, otherwise
$K\subseteq N(x)$ and $S\subseteq V(T)\setminus N(x)$, hence $K$ and $S$ are separated by a cut of $F'$.
 Otherwise, by property of $\mathcal{F}$, $K'$ and $S'$ intersect on a switchable pair $xy$: then the same argument can be applied with $U_{i,}$ for some $i\in\set{1, 2, 3, 4}$ depending on the intersection between $\{x, y\}$ and $K$.
\end{proof}

In particular, as for the graph case, if $T\in \mathcal{F}$ has at most $\mathcal{O}(|V(T)|^c)$ maximal cliques (or stable sets) for some constant $c\geq 2$, then there is a CS-separator of size $\mathcal{O}(|V(T)|^c)$.

\section{Decomposing trigraphs of $\mathcal{F}$}
\label{sec: decomposition}
This section recalls definitions and results from \cite{bergefreebsp}
that we use in the next section. Our goal is to state the
decomposition theorem for trigraphs of $\mathcal{F}$ and to define the blocks of decomposition.
First we need some definitions.

\subsection{Basic trigraphs}

We need the counterparts of \emph{bipartite graphs} (and their complements), \emph{line graphs of bipartite graphs} (and their complements), and \emph{double-split graphs} which are the basic classes for decomposing Berge graphs. For the trigraph case, the basic classes are \emph{bipartite trigraphs} and their complements, \emph{line trigraphs} and their complements, and \emph{doubled trigraphs}.

A trigraph $T$ is \defi{bipartite} if its vertex set can be partitioned 
into two strong stable sets.
A trigraph $T$ is a \defi{line trigraph} if the full realization of $T$
is the line graph of a bipartite graph and every clique of size at
least $3$ in $T$ is a strong clique.
Let us now define the  analogue
of the double split graph, namely the
doubled trigraph.  A \defi{good partition} of a trigraph $T$ is
a partition $(X, Y)$ of $V(T)$ (possibly, $X=\emptyset$ or
$Y=\emptyset$) such that:

\begin{itemize}
\item Every  component of $T[X]$ has at most two vertices, and every
  anticomponent of $T[Y]$ has at most two vertices.
\item No switchable pair of $T$ meets both $X$ and $Y$. 
\item For every component $C_X$ of $T[X]$, every anticomponent $C_Y$ of
  $T[Y]$, and every vertex $v$ in $C_X \cup C_Y$, there exists at most
  one strong edge and at most one strong antiedge between $C_X$ and
  $C_Y$ that is incident to $v$.
\end{itemize}

A trigraph is \defi{doubled} if it has a good partition.
A trigraph is \defi{basic} if it is either a bipartite trigraph, the
complement of a bipartite trigraph, a line trigraph, the complement of
a line trigraph or a doubled trigraph.
Basic trigraphs behave well with respect to induced subtrigraphs and
complementation as stated by the following lemma.

\begin{lemma}[\cite{bergefreebsp}]
  Basic trigraphs are Berge and are closed under taking induced
  subtrigraphs, semirealizations, realizations and complementation.
\end{lemma}

\subsection{Decompositions}

We now describe the decompositions that we need for the
decomposition theorem. They generalize the
decompositions used in the Strong Perfect Graph Theorem \cite{SPGT}, and in addition all the
important crossing edges and non-edges in those graph decompositions are
required to be strong edges and strong antiedges of the trigraph,
respectively.

First, a \defi{$2$-join} in a trigraph $T$ (see Figure \ref{fig: 2-join trigraphs}.(a) for an illustration) is a
partition $(X_1, X_2)$ of $V(T)$ such that there exist disjoint sets $A_1, B_1, C_1, A_2, B_2, C_2\subseteq V(T)$ satisfying:

\begin{itemize}
\item $X_1=A_1\cup B_1\cup C_1$ and $X_2=A_2\cup B_2\cup C_2$.
\item $A_1, A_2, B_1$ and $B_2$ are non-empty.
\item No switchable pair meets both $X_1$ and $X_2$.
\item Every vertex of $A_1$ is
strongly adjacent to every vertex of $A_2$, and every vertex of $B_1$ is
strongly adjacent to every vertex of $B_2$.
\item There are no other strong edges between $X_1$ and $X_2$.
\item For $i=1,2$ $|X_i| \geq 3$. 
\item For $i = 1,2$, if $|A_i| = |B_i| = 1$, then the full realization of
$T[X_i]$ is not a path of length two joining the members of $A_i$ and $B_i$.
\item For $i = 1,2$, every component of $T[X_i]$ meets both $A_i$ and
$B_i$ (this condition is usually required only for a \emph{proper} 2-join, but we will only deal with proper 2-join in the following).
\end{itemize}

 A \defi{complement $2$-join} of a
trigraph $T$ is a $2$-join in $\overline{T}$.
When proceeding by induction on the number of vertices, we sometimes want to contract one side of a 2-join into three vertices and assert that the resulting trigraph is smaller. This does not come directly from the definition (we assume only $|X_i|\geq 3$), but can be deduced from the following technical lemma:

\begin{lemma}[\cite{bergefreebsp}]
\label{2joinform}
  Let $T$ be a trigraph from $\mathcal F$ with no balanced
  skew-partition, and let $(X_1, X_2)$ be a 
  $2$-join in $T$. Then $|X_i| \geq 4$, for $i=1,2$.
\end{lemma}

Moreover, when decomposing by a 2-join, we need to be careful about the parity of the lengths of paths from $A_i$ and $B_i$ in order not to create an odd hole. In this respect, the following lemma is useful:

\begin{lemma}[\cite{bergefreebsp}]
  Let $T$ be a Berge trigraph and $(A_1, B_1, C_1,
A_2, B_2, C_2)$ a split of a $2$-join of $T$. Then all paths
with one end in $A_i$, one end in $B_i$ and interior in $C_i$, for
$i=1, 2$, have lengths of the same parity.
\end{lemma}

\begin{proof} Otherwise, for $i=1, 2$, let $P_i$ be a path with one end in
$A_i$, one end in $B_i$ and interior in $C_i$, such that $P_1$ and
$P_2$ have lengths of different parity. They form an odd hole, a
contradiction.
\end{proof}

Consequently, a $2$-join in a Berge trigraph is said \defi{odd} or \defi{even} according
to the parity of the lengths of the paths between $A_i$ and $B_i$. The
 lemma above ensures the correctness of the definition.

Our second decomposition is the balanced skew-partition. A \defi{skew-partition} is a partition $(A,B)$ of $V(T)$ such that $A$ is not
connected and $B$ is not anticonnected. 
It is moreover \defi{balanced} if there is no odd path of length greater than $1$ with
ends in $B$ and interior in $A$, and there is no odd antipath of
length greater than $1$ with ends in $A$ and interior in $B$.

\begin{figure}
\center
\subfigure[A 2-join.]{\includegraphics[scale=1, page=1]{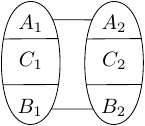}}
\hspace{30pt}
\subfigure[Block of decomposition $T_{X_1}$ for an odd 2-join.]{\hspace{20pt}\includegraphics[scale=1, page=3]{fig/BlockDecompo2Join}\hspace{20pt}}
\hspace{30pt}
\subfigure[Block of decomposition $T_{X_1}$ for an even 2-join.]{\hspace{20pt}\includegraphics[scale=1, page=2]{fig/BlockDecompo2Join}\hspace{20pt}}

\caption{Diagram for a 2-join and its blocks of decomposition. Straight lines stand for strongly complete sets, and wiggly edges stand for switchable pairs. No other edge can cross between left and right.}
\label{fig: 2-join trigraphs}
\end{figure}

We are now ready to state the decomposition theorem.

\begin{theorem}[\cite{bergefreebsp}, adapted from \cite{TheseMaria}]
  \label{structure}
  Every trigraph in $\mathcal{F}$ is either basic, or admits a
  balanced skew-partition, a $2$-join, or a complement $2$-join.
\end{theorem}

We now define the
\emph{blocks of decomposition} $T_{X_1}$ and $T_{X_2}$ of a 2-join $(X_1, X_2)$ in a trigraph $T$ (an illustration of blocks of decomposition can be found in Figure \ref{fig: 2-join trigraphs}). Let $(A_1, B_1, C_1, A_2, B_2, C_2)$ be a split of $(X_1, X_2)$. Informally, the block $T_{X_1}$ is obtained from $T$ by keeping $X_1$ as it is and contracting $X_2$ into few vertices, depending on the parity of the 2-join: 2 vertices for odd 2-joins (one for $A_2$, one for $B_2$), and 3 vertices for even 2-joins (one extra-vertex for $C_2$). 

If the 2-join is odd, we build the
block of decomposition $T_{X_1}$ as follows: we start with $T [
A_1\cup B_1\cup C_1]$. We then add two new \defi{marker vertices} $a_2$
and $b_2$ such that $a_2$ is strongly complete to $A_1$, $b_2$ is strongly
complete to $B_1$, $a_2b_2$ is a switchable pair, and there are no other 
edges between $\{a_2,b_2\}$ and $X_1$.  Note that $\{a_2, b_2\}$
is a switchable component of $T_{X_1}$. The block of decomposition $T_{X_2}$ is defined similarly with marker vertices $a_1$ and $b_1$.

If the 2-join is even, we build the
block of decomposition $T_{X_1}$ as follows: once again, we start with $T[
A_1\cup B_1\cup C_1]$. We then add three new \defi{marker vertices}
$a_2$, $b_2$ and $c_2$ such that $a_2$ is strongly complete to $A_1$, $b_2$ is
strongly complete to $B_1$, $a_2c_2$ and $c_2b_2$ are switchable pairs, and
there are no other edges between $\{a_2, b_2, c_2\}$ and $X_1$. The block of decomposition $T_{X_2}$ is defined similarly with marker vertices $a_1$, $b_1$ and $c_1$.

We define the blocks of decomposition of a complement $2$-join $(X_1,X_2)$ in $T$ as the
complement of the blocks of decomposition of the $2$-join $(X_1,X_2)$
in $\comp{T}$.

The following theorem ensures that the blocks of decomposition do not leave
the class:
 
\begin{theorem}[\cite{bergefreebsp}]
 \label{l:stayBerge}
 If $(X_1, X_2)$ is a $2$-join or a complement $2$-join of a
 trigraph $T$ from $\mathcal F$ with no balanced skew-partition, then
 $T_{X_1}$ and $T_{X_2}$ are trigraphs from~$\mathcal F$ with no balanced skew-partition.
\end{theorem}

Observe that this property is essential to apply the induction hypothesis when contracting a 2-join or complement 2-join.
This is what trigraphs are useful for: putting a strong edge or a strong antiedge instead of a switchable pair in the blocks of decomposition may create a balanced skew-partition.

\section{Proving the Clique-Stable Set separation}
\label{sec: CS-sep}

\subsection{In Berge graphs with no balanced skew-partition}

This part is devoted to proving that  trigraphs of $\mathcal{F}$ with no balanced skew-partition admit a quadratic CS-separator. 
The result is proved by induction, and so there are two cases to consider: either the trigraph is basic (handled in Lemma \ref{lemma:basic}); or the trigraph, or its complement can be decomposed by a 2-join (handled in Lemma \ref{2-join}). We put the pieces together in Theorem \ref{th: CS-Sep perfect sans skew}.

We begin with the case of basic trigraphs:

\begin{lemma}\label{lemma:basic}
There exists a constant $c$ such that every basic trigraph $T$ admits a \mbox{CS-separator} of size $c|V(T)|^2$.
\end{lemma}

\begin{proof}
Since the problem is self-complementary, we consider only the cases of bipartite trigraphs, line trigraphs and doubled trigraphs.
A clique in a bipartite trigraph has size at most 2, thus there is at most a quadratic number of them. If $T$ is a line trigraph, then its full realization is the line graph of a bipartite graph $G$ thus $T$ has a linear number of maximal cliques (each of them corresponds to a vertex of $G$). By Observation \ref{clique max trigraph}, this implies the existence of a CS-separator of quadratic size.

If $T$ is a doubled trigraph, let $(X,Y)$ be a good partition of $T$ and consider the following family of cuts:
first, build the cut $(Y,X)$, and in the second place, for every
$Z=\{x\}$ with $x\in X$ or $Z=\emptyset$, and for every $Z'=\{y\}$ with $y\in Y$ or $Z'=\emptyset$,
build the cut $((Y \cup Z)\setminus Z', (X\cup Z') \setminus
Z)$. Finally, for every pair $x,y\in V$, build the cut $(\{x,y\},
V(T)\setminus \{x,y\})$, and $(V(T)\setminus \{x,y\}, \{x,y\})$. These cuts
form a CS-separator : let $K$ be a clique in $T$ and $S$ be a stable set disjoint from $K$, then $|K\cap X|\leq 2$ and $|S\cap Y|\leq
2$. If $|K\cap X|= 2$, then $K$ has size exactly 2 since no vertex of $Y$ has two adjacent neighbors in $X$. So the cut $(K,
V\setminus K)$ separates $K$ and $S$. By similar arguments, if $|S\cap Y|=2$ then $S$ has size 2 and $(V\setminus S, S)$ separates $K$ and $S$.
Otherwise, $|K\cap
X|\leq 1$ and $|S\cap Y|\leq 1$ and then $(Y\cup (K\cap X) \setminus
(S\cap Y), X\cup (S\cap Y) \setminus (K\cap X))$ separates $K$ and $S$.
\end{proof}

Next, we handle the case where a $2$-join appears in the trigraph and show how to reconstruct a CS-separator from the CS-separators of the blocks of decompositions.

\begin{lemma}\label{2-join}
Let $T$ be a trigraph admitting a $2$-join $(X_1, X_2)$. If the blocks of decomposition $T_{X_1}$ and $T_{X_2}$ admit a CS-separator of size respectively $k_1$ and $k_2$, then $T$ admits a CS-separator of size $k_1+k_2$.
\end{lemma}

\begin{proof} 
 Let $(A_1, B_1, C_1, A_2, B_2, C_2)$ be a split of $(X_1, X_2)$, $T_{X_1}$ (resp. $T_{X_2}$) be the block of decomposition with marker vertices $a_2, b_2$, and possibly $c_2$ (depending on the parity of the $2$-join) (resp. $a_1, b_1$, and possibly $c_1$). Observe that there is no need to distinguish between an odd or an even $2$-join, because $c_1$ and $c_2$ play no role. Let $F_1$ be a CS-separator of $T_{X_1}$ of size $k_1$ and $F_2$ be a CS-separator of $T_{X_2}$ of size $k_2$.
 
Let us build $F$ aiming at being a CS-separator for $T$. For each cut $(U,W)\in F_1$,
build a cut as follows: start with $U'=U\cap X_1$ and $W'=W\cap X_1$. If $a_2\in U$, add $A_2$ to $U'$, otherwise add $A_2$ to $W'$. Moreover if $b_2\in U$, add $B_2$ to $U'$, otherwise add $B_2$ to $W'$. Now build the cut $(U',W'\cup C_2)$ with the resulting sets $U'$ and $W'$.
 In other words, we put $A_2$ on the same side as $a_2$, $B_2$ on the same side as $b_2$, and $C_2$ on the stable set side. For each cut $(U,W)$ in $F_2$, we do the similar construction: start from $(U\cap X_2, W\cap X_2)$, then put $A_1$ on the same side as $a_1$, $B_1$ on the same side as $b_1$, and finally put $C_1$ on the stable set side.

$F$ is indeed a CS-separator: let $K$ be a clique and $S$ be a stable set disjoint from $K$. First, suppose that $K\subseteq X_1$. We define $S'=(S\cap X_1) \cup S_{a_2,b_2}$ where $S_{a_2,b_2}\subseteq\{a_2, b_2\}$ contains $a_2$ (resp. $b_2$) if and only if $S$ intersects $A_2$ (resp. $B_2$). $S'$ is a stable set of $T_{X_1}$, so there is a cut in $F_1$ separating the pair $K$ and $S'$. The corresponding cut in $F$ separates $K$ and $S$. The case $K\subseteq X_2$ is handled symmetrically.

Finally, suppose $K$ intersects both $X_1$ and $X_2$. Then 
$K\cap C_1=\emptyset$ and $K\subseteq A_1\cup A_2$ or $K\subseteq B_1\cup B_2$. Assume by symmetry that $K\subseteq A_1\cup A_2$.
Observe that $S$ can not intersect both $A_1$ and $A_2$ which are strongly complete to each other, so without loss of generality assume it does not intersect $A_2$.
 Let $K'=(K\cap A_1)\cup \{a_2\}$ and $S'=(S\cap X_1)\cup S_{b_2}$ where $S_{b_2}=\{b_2\}$ if $S$ intersects $B_2$, and $S_{b_2}=\emptyset$ otherwise. $K'$ is a clique and $S'$ is a stable set of $T_{X_1}$ so there exists a cut in $F_1$ separating them, and the corresponding cut in $F$ separates $K$ and $S$. Then $F$ is a CS-separator.
\end{proof}

This leads us to the main theorem of this section:

\begin{theorem}
\label{th: CS-Sep perfect sans skew}
Every trigraph $T$ of $\mathcal{F}$ with no balanced skew-partition admits a CS-separator of size $\mathcal{O}(|V(T)|^2)$.
\end{theorem}

\begin{proof}
Let $c'$ be the constant of Lemma \ref{lemma:basic} and $c=\max(c',2^{24})$. Let us prove by induction that every trigraph of $T$ on $n$ vertices admits a CS-separator of size $cn^2$.
The initialization is concerned with basic trigraphs, for which Lemma \ref{lemma:basic} shows that a CS-separator of size $c'n^2$ exists, and with trigraphs of size less than $24$. For them, one can consider every subset $U$ of vertices and take the cut $(U, V\setminus U)$ which form a trivial CS-separator of size at most $2^{24}n^2$.

Consequently, we can now assume that the trigraph $T$ is not basic and has at least $25$ vertices. By applying Theorem \ref{structure}, we know that $T$ has a $2$-join $(X_1, X_2)$ (or a complement $2$-join, in which case we switch to $\overline{T}$ since the problem is self-complementary). We define $n_1=|X_1|$, then by Lemma \ref{2joinform} we can assume that $4\leq n_1\leq n-4$. Applying Theorem \ref{l:stayBerge}, we can apply the induction hypothesis on the blocks of decomposition $T_{X_1}$ and $T_{X_2}$ to get a CS-separator of size respectively at most $k_1=c(n_1+3)^2$ and $k_2=c(n-n_1+3)^2$. By Lemma \ref{2-join}, $T$ admits a CS-separator of size $k_1+k_2$. The goal is to prove that $k_1+k_2\leq cn^2$.

Let $P(n_1)=c(n_1+3)^2+c(n-n_1+3)^2-cn^2$. $P$ is a degree $2$ polynomial with leading coefficient $2c>0$. Moreover, $P(4)=P(n-4)=-2c(n-25)\leq 0$ so by convexity of $P$, $P(n_1)\leq 0$ for every $4\leq n_1 \leq n-4$, which achieves the proof.
\end{proof}

\subsection{Closure by generalized $k$-join}
\label{sec: k-join}
We present here a way to extend the result of the Clique-Stable Set separation on Berge graphs with no balanced skew-partition to larger classes of graphs, based on a generalization of the $2$-join.
Let $\mathcal{C}$ be a class of graphs, which should be seen as ``basic'' graphs. For any integer $k\geq 1$, we construct the class $\mathcal{C}^{\leq k}$ of trigraphs in the following way: a trigraph $T$ belongs to $\mathcal{C}^{\leq k}$ if and only if there exists a partition $X_1, \ldots, X_r$ of $V(T)$ such that:

\begin{itemize}
\item For every $1\leq i \leq r$, $1\leq |X_i|\leq k$.
\item For every $1\leq i \leq r$, ${\choose{2}{X_i}}\subseteq \sigma(T)$.
\item For every $1\leq i\neq j\leq r$, $(X_i\times X_j) \cap \sigma(T) =\emptyset$.
\item There exists a graph $G$ in $\mathcal{C}$ such that $G$ is a realization of $T$.
\end{itemize}

In other words, starting from a graph $G$ of $\mathcal{C}$, we
partition its vertices into small parts (of size at most $k$), and change all
adjacencies inside the parts into switchable pairs.

We now define the \defi{generalized $k$-join} between two trigraphs
$T_1$ and $T_2$ (see Figure \ref{fig:k-join} for an illustration), which generalizes the $2$-join and is quite similar to the $H$-join defined in \cite{BuiXuan2010}. Let $T_1$ and $T_2$ be two trigraphs having the following properties, with $1\leq r,s\leq k$:

\begin{itemize}
\item $V(T_1)$ is partitioned into $(A_1, \ldots, A_r, B=\{b_1, \ldots, b_s\})$  and $A_j\neq \emptyset$ for every $1\leq j\leq r$.
\item $V(T_2)$ is partitioned into $(B_1, \ldots, B_s, A=\{a_1, \ldots,
  a_r\})$ and $B_i\neq \emptyset$ for every $1\leq i \leq s$.
\item ${\choose{2}{B}}\subseteq \sigma(T_1)$ and ${\choose{2}{A}}\subseteq \sigma(T_2)$, meaning that $A$ and $B$ contain only switchable pairs.
\item For every $1\leq i\leq s, 1\leq j \leq r$, $b_i$ and $a_j$ are either both strongly complete or both strongly anticomplete to respectively $A_j$ and $B_i$. In other words, there exists a bipartite graph describing the adjacency between $B$ and $(A_1, \ldots, A_r)$, and the same bipartite graph describes the adjacency between $(B_1, \ldots, B_s)$ and $A$.
\end{itemize}

Then the generalized $k$-join of $T_1$ and $T_2$ is the trigraph $T$ with vertex set
 $V(T)=A_1\cup \ldots \cup A_r\cup B_1\cup \ldots\cup B_s$. Let
$\theta_1$ and $\theta_2$ be the adjacency functions of $T_1$ and
$T_2$, respectively. As much as possible, the adjacency function
$\theta$ of $T$ follows $\theta_1$ and $\theta_2$ (meaning
$\theta(uv)=\theta_1(uv)$ for $uv\in {\choose{2}{V(T_1)\cap V(T)}}$
and $\theta(uv)=\theta_2(uv)$ for $uv\in {\choose{2}{V(T_2)\cap V(T)}}$), and for $a\in A_j$, $b\in B_i$, $\theta(ab)=1$ if $b_i$ and
$A_j$ are strongly complete in $T_1$ (or, equivalently, if $a_j$ and
$B_i$ are strongly complete in $T_2$), and $-1$ otherwise.

\begin{figure}

\centering
\subfigure[In $T_1$, $b_1b_2$ is a switchable pair, $b_1$ is strongly complete to $A_1$ and $A_2$ and strongly anticomplete to $A_3$; $b_2$ is strongly complete to $A_2$ and $A_3$ and strongly anticomplete to $A_1$. There can be any adjacency in the left part.]{
 \includegraphics[scale=0.9]{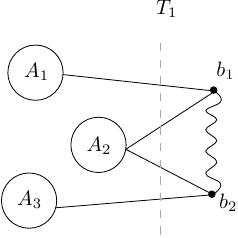}}
 \hspace{15pt}
 \subfigure[In $T_2$, $\{a_1,a_2, a_3\}$ contains only switchable pairs, $B_1$ is strongly complete to $\{a_1, a_2\}$ and strongly anticomplete to $a_3$; $B_2$ is strongly complete to $\{a_2, a_3\}$ and  strongly anticomplete to $a_1$. There can be any adjacency in the right part.]{
 \includegraphics[scale=0.9]{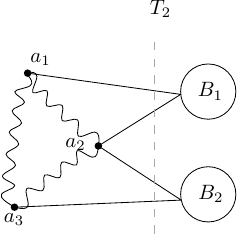}}
  \hspace{15pt}
 \subfigure[In $T$, $B_1$ is strongly complete to $A_1$ and $A_2$ and strongly anticomplete to $A_3$; $B_2$ is strongly complete to $A_2$ and $A_3$ and strongly anticomplete to $A_1$. The adjacencies inside the left part and the right part are preserved.]{
 \includegraphics[scale=0.9]{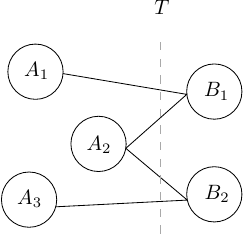}}

 \caption{Example of a generalized 3-join $T$ of $T_1$ and $T_2$ with $r=3$ and $s=2$.}
  \label{fig:k-join}
\end{figure}

We finally define $\overline{\mathcal{C}^{\leq k}}$ to be the smallest class of trigraphs containing $\mathcal{C}^{\leq k}$ and closed under generalized $k$-join.

\begin{lemma}\label{CS-sep ck intermediaire}
If every graph $G$ of $\mathcal{C}$ admits a CS-separator of size
$m$, then every trigraph $T$ of $\mathcal{C}^{\leq k}$ admits a CS-separator of size $m^{k^2}$.
\end{lemma}

\begin{proof}
  First we claim that if there exists a CS-separator $F$ of size
  $m$ then the family of cuts $F'=\{(\cap_{i=1}^kU_i,
  \cup_{i=1}^kW_i)|(U_1,W_1),\ldots , (U_k, W_k) \in F\}$ has size $m^k$ and separates every clique from
  every union of at most $k$ stable sets. Indeed if $K$ is a clique and
  $S_1,\ldots ,S_k$ are $k$ stable sets disjoint from $K$
  then there exist in $F$ $k$ partitions $(U_1,W_1), \ldots , (U_k, W_k)$ such that $(U_i, W_i)$ separates $K$ and $S_i$.
Now $(\cap_{i=1}^kU_i, \cup_{i=1}^kW_i)$ is a partition that
  separates $K$ from $\cup_{i=1}^k S_i$. Using the same argument we
  can build a family  of cuts $F''$ of size $m^{k^2}$
  that separates every union of at most $k$ cliques from every union of
  at most $k$ stable sets. Now let $T$ be a trigraph of $\mathcal{C}^{\leq k}$ and let  $G\in\mathcal{C}$ such that $G$ is a realization of $T$. Let $X_1, \ldots, X_r$ be the partition of $V(T)$ as in the definition of $\mathcal{C}^{\leq k}$. Notice that a clique $K$ (resp.~stable set $S$) in
  $T$ is a union of at most $k$ cliques (resp.~stable sets) in $G$: indeed, by taking one vertex in $K\cap X_i$  (if not empty) for each $1\leq i \leq r$, we build a clique of $G$; repeating this operation at most $k$ times covers $K$ with $k$ cliques of $G$. It follows that there exists a CS-separator of $T$ of size $m^{k^2}$.
\end{proof}

\begin{lemma}\label{k-join}
If $T_1$, $T_2$ $\in \overline{\mathcal{C}^{\leq k}}$ admit CS-separators of size respectively $m_1$ and $m_2$, then the generalized $k$-join $T$ of $T_1$ and $T_2$ admits a CS-separator of size $m_1+m_2$.
\end{lemma}

\begin{proof}
The proof is very similar to the one of Lemma \ref{2-join}. We follow the notation introduced in the definition of the generalized $k$-join. Let $F_1$ (resp. $F_2$) be a CS-separator of size $m_1$ (resp. $m_2$) on $T_1$ (resp. $T_2$). Let us build $F$ aiming at being a CS-separator on $T$. For every cut $(U,W)$ in $F_1$, build the cut $(U',W')$ with the following process: 
start with $U'=U\cap \cup_{j=0}^{r}A_j$ and $W'= W \cap \cup_{j=0}^{r}A_j$; now for every $1\leq i \leq s$, if $b_i\in U$, then add $B_i$ to $U'$, otherwise add $B_i$ to $W'$. In other words, we take a cut similar to $(U,W)$ by putting $B_i$ in the same side as $b_i$. We do the symmetric operation for every cut $(U,W)$ in $F_2$ by putting $A_j$ in the same side as $a_j$. 

$F$ is indeed a CS-separator: let $K$ be a clique and $S$ be a stable set disjoint from $K$. Suppose as a first case that one part of the partition $(A_1, \ldots, A_r, B_1, \ldots, B_s)$ intersects both $K$ and $S$. Without loss of generality, we assume that $A_1\cap K\neq \emptyset$ and $A_1 \cap S \neq \emptyset$. Since for every $i$, $A_1$ is either strongly complete or strongly anticomplete to $B_i$, $B_i$ can not intersect both $K$ and $S$. Consider the following sets in $T_1$: $K'=(K\cap V(T))\cup K_b$ and $S'=(S\cap V(T))\cup S_b$ where \mbox{$K_b=\{b_i| K\cap B_i \neq \emptyset\}$} and $S_b=\{b_i | S\cap B_i \neq \emptyset\}$. $K'$ is a clique in $T_1$, $S'$ is a stable set in $T_1$, and there is a cut separating them in $F_1$. The corresponding cut in $F$ separates $K$ and $S$.

In the case when no part of the partition intersects both $K$ and $S$, analogous argument applies.
\end{proof}

\begin{theorem}
If every graph $G\in \mathcal{C}$ admits a CS-separator of size $\mathcal{O}(|V(G)|^c)$, then every trigraph $T\in \overline{\mathcal{C}^{\leq k}}$ admits a CS-separator of size $\mathcal{O}(|V(T)|^{k^2c})$. In particular, every realization $G'$ of a trigraph of $\overline{\mathcal{C}^{\leq k}}$ admits a CS-separator of size $\mathcal{O}(|V(G')|^{k^2c})$.
\end{theorem}

\begin{proof}
Let $p'$ be
  a constant such that every $G\in \mathcal{C}$ admits a   \mbox{CS-separator} of size $p'|V(G)|^c$,
  and let $p_0$  be a large constant to be defined
  later.
  We prove by induction that there exists a
  CS-separator of size $pn^{k^2c}$ with $p=\max(p',2^{p_0})$. The base case is divided into two cases: the trigraphs of
  $\mathcal{C}^{\leq k}$, for which the property is verified according to Lemma \ref{CS-sep ck intermediaire};
  and the trigraphs of size at most $p_0$, for which one can consider
  every subset $U$ of vertices and take the cut $(U, V\setminus U)$
  which form a trivial CS-separator of size at most $2^{p_0}n^{k^2c}$.

  Consequently, we can now assume that $T$ is the generalized $k$-join
  of $T_1$ and $T_2$ with at least $p_0$ vertices. Let $n_1=|T_1|$ and
  $n_2=|T_2|$ with $n_1+n_2=n+r+s$ and $r+s+1\leq n_1, n_2, \leq
  n-1$. By induction, there exists a CS-separator of size
  $pn_1^{k^2c}$ on $T_1$ and one of size $pn_2^{k^2c}$ on $T_2$. By
  Lemma \ref{k-join}, there exists a CS-separator on $T$ of size
  $pn_1^{k^2c}+pn_2^{k^2c}$. The goal is to prove
  $pn_1^{k^2c}+pn_2^{k^2c}\leq pn^{k^2c}$.

Notice that $n_1+n_2=n-1+r+s+1$ so by convexity of $x \mapsto x^c$ on $\mathbb{R}^+$,
$n_1^{k^2c}+n_2^{k^2c}\leq (n-1)^{k^2c}+(r+s+1)^{k^2c}$. Moreover, $r+s+1\leq
2k+1$. Now we can define $p_0$ large enough such that for every $n\geq
p_0$, $n^{k^2c}-(n-1)^{k^2c}\geq (2k+1)^{k^2c}$. Then $n_1^{k^2c}+n_2^{k^2c}\leq
n^{k^2c}$, which concludes the proof.
\end{proof}

\section{Strong \EH \ property}
\label{sec: SEH}

As mentioned in the introduction, a \emph{biclique} in $T$ is a pair $(V_1, V_2)$ of disjoint subsets of vertices such that $V_1$ is strongly complete to $V_2$. Observe that we do not care about the inside of $V_1$ and $V_2$. The size of the biclique is $\min(|V_1|, |V_2|)$.
A \emph{complement biclique} in $T$ is a biclique in $\comp{T}$.
Let $\mathcal{C}$ be a class of trigraphs, then we say that $\mathcal{C}$ has the \emph{Strong \EH \ property} if there exists $c>0$ such that for every $T\in \mathcal{C}$, $T$ admits a biclique or a complement biclique of size at least $c|V(T)|$. This notion was introduced by Fox and Pach \cite{FoxPach08}, and they proved that if a hereditary class of graphs $\mathcal{C}$  has the Strong \EH \ property, then it has the \EH \ property. Moreover, it was proved in \cite{Bousquet13} that, under the same assumption, there exists $c>0$ such that every graph $G\in \mathcal{C}$ admits a CS-separator of size $\mathcal{O}(|V(G)|^c)$. However, the class of trigraphs of $\mathcal{F}$ with no balanced skew-partition is not hereditary so we can not apply this here. The goal of Subsection \ref{subsec: SEH no BSP} is to prove the following theorem, showing that the Strong \EH \ property holds for the class of trigraphs under study:

\begin{theorem}
\label{th: SEH berge no BSP}
Let $T$ be a trigraph of $\mathcal{F}$ with no balanced skew-partition. If $|V(T)|\geq 3$, then $T$ admits a biclique or a complement biclique of size at least $|V(T)|/55$.
\end{theorem}

 \subsection{In Berge trigraphs with no balanced skew-partition}
 \label{subsec: SEH no BSP}

\renewcommand{\wr}{w_r} 
\newcommand{\wc}{w_c} 
\newcommand{\wt}{w_t} 
\newcommand{\wac}{w_{\comp{c}}} 
\newcommand{\wrp}{w'_r} 
\newcommand{\wcp}{w'_c} 
\newcommand{\wtp}{w'_t} 
\newcommand{\wacp}{w'_{\comp{c}}} 
\newcommand{\br}{\beta_r} 
\newcommand{\bc}{\beta_c} 
\newcommand{\bt}{\beta_t} 
\newcommand{\bac}{\beta_{\comp{c}}} 
\newcommand{\brp}{\beta'_r} 
\newcommand{\bcp}{\beta'_c} 
\newcommand{\btp}{\beta'_t} 
\newcommand{\bacp}{\beta'_{\comp{c}}} 
\newcommand{\D}[1][1]{\frac{#1}{55} \cdot }

We need a weighted version in order for the proof to work. When one faces a 2-join with split $(A_1, B_1, C_1, A_2, B_2, C_2)$, the idea is to contract $A_i$, $B_i$, and $C_i$ for $i=1$ or $2$, with the help of the blocks of decomposition, until we reach a basic trigraph. 
The weight is meant for keeping track of the contracted vertices.
We then find a  biclique (or complement biclique) of large weight in the basic trigraph, because it is well-structured,  and we prove that we can backtrack and transform it into a biclique (or complement biclique) in the original trigraph.

However, this sketch of proof is too good to be true:
 in case of an odd 2-join or odd complement 2-join $(X_1, X_2)$  with split $(A_1, B_1, C_1, A_2, B_2, C_2)$, the block of decomposition $T_{X_1}$ does not contain any vertex that stands for $C_2$. Thus we have to put the weight of $C_2$ on the switchable pair $a_2b_2$, and remember whether $C_2$ was strongly anticomplete (in case of a 2-join) or strongly complete (in case of a complement 2-join) to $X_1$. This may propagate if we further contract $a_2b_2$.

 Let us now introduce some formal notation. A \defi{weighted trigraph} is a pair $(T,w)$ where $T$ is a trigraph and $w$ is a weight function which assigns:
\begin{itemize}
 \item to every vertex $v\in V(T)$,  a triple $w(v)=(\wr(v), \wc(v), \wac(v))$. 
\item to every switchable pair $uv\in \sigma(T)$, a pair $w(uv)=(\wc(uv), \wac(uv))$. 
 \end{itemize}
 In both cases, each coordinate has to be a non-negative integer. For $v\in V$, $\wr(v)$ is called the \defi{real weight} of $v$, and for $x\in V$ or $x\in \sigma(T)$,  $\wc(x)$ (resp. $\wac(x)$) is called the \defi{extra-complete} (resp. \defi{extra-anticomplete}) weight of $x$. 
 The extra-anticomplete (resp. extra-complete) weight will stand for vertices that have been deleted during the decomposition of an odd 2-join (resp. odd complement 2-join) - the $C_2$ in the discussion above - and thus which were strongly anticomplete (resp. strongly complete) to the other side of the 2-join. 
 
 Let us mention the some further notation: given a set of vertices $U\subseteq V(T)$, the weight of $U$ is $w(U) =(\wr(U), \wc(U), \wac(U))$ where $\wr(U)$ is the sum of $\wr(v)$ over all $v\in U$, and
 $$\wc(U)  = \sum_{\substack{u,v\in U \\ uv\in \sigma(T)}} \wc(uv)+\sum_{v\in U} \wc(v)  \quad \text{and} \quad
 \wac(U)  = \sum_{\substack{u,v\in U \\ uv\in \sigma(T)}} \wac(uv)+\sum_{v\in U} \wac(v) \ . $$
 
The \defi{total weight} of $U$ is $\wt(U)=\wr(U)+\wc(U)+\wac(U)$.  
 We abuse  notation and write $w(T)$ instead of $w(V(T))$, and in particular the total weight of $T$ will be denoted $\wt(T)$.
 Given two disjoint sets of vertices $A$ and $B$, the \defi{crossing weight} $w(A,B)$ is defined as the weight of the switchable pairs with one endpoint in $A$ and the other in $B$, namely $w(A,B)=(\wc(A,B), \wac(A,B))$ where $\wc(A,B)$ (resp. $\wac(A,B)$) is the sum of $\wc(ab)$ (resp. $\wac(ab)$) over all $a\in A$, $b\in B$ such that $ab\in \sigma(T)$.

  An unfriendly behavior for a weight function is to  concentrate all the weight at the same place, or to have a too heavy extra-complete and extra-anticomplete weight, this is why we introduce the following. A weight function $w$ is \defi{balanced} if the following conditions hold:
  
  \begin{itemize}
  \item For every $v\in V(T)$, $\wr(v)\leq \D \wt(T)$.
  \item For every $x\in V(T)$ or $x\in \sigma(T)$, $\max(\wc(x), \wac(x))\leq \D \wt(T)$.
  \item $\wc(T)+\wac(T)\leq \D[7] \wt(T)$.
  \end{itemize}

 A \defi{virgin weight} on $T$ is a weight  $w$ such that ${w}_c(T)={w}_{\comp{c}}(T)=0$. In such a case, we will drop the subscript and simply denote $w(v)$ for ${w}_r(v)$.
 The \emph{weight} of a biclique (or complement biclique) $(X,Y)$ is $\min(\wr(X), \wr(Y))$.
From now on, the goal is to find a biclique or a complement biclique of large weight, that is to say a constant fraction of $\wt(T)$.
  
  We need a few more definitions, concerning in particular 
how to adapt the blocks of decomposition to the weighted setting.  
 Let $(T,w)$ be a weighted trigraph
such that $T$ admits a $2$-join or complement 2-join $(X_1, X_2)$ with split $(A_1, B_1, C_1,
A_2, B_2, C_2)$. Without loss of generality, we can assume that $X_1$ is the ``heavier'' part, \emph{i.e.} $\wt(X_1)\geq \wt(T)/2$. The \defi{contraction of $(T,w)$} (with respect to this split) is the weighted trigraph
$(T',w')$, where $T'$ is the block of decomposition $T_{X_1}$ and where $w'$ is defined as follows:
\begin{itemize}
\item For every vertex $v\in X_1$, we define $w'(v)=w(v)$.
\item For marker vertices $a_2$ and $b_2$, we set $w'(a_2)=w(A_2)$ and $w'(b_2)=w(B_2)$.
\item In case of an even (complement or not) 2-join, we have $w'(c_2)=w(C_2)$, $w'(a_2c_2)=w(A_2,C_2)$ and $w'(b_2c_2)=w(B_2,C_2)$.
\item In case of an odd 2-join, the marker vertex $c_2$ does not exist so things become slightly more complicated: since we want to preserve the total weight, the switchable pair $a_2b_2$ has to take a lot of weight, including the real weight of $C_2$; $\wr(C_2)$ is thus given as an extra-anticomplete weight to $a_2b_2$ because $C_2$ is strongly anticomplete to every other vertex outside of $A_2\cup B_2$.
For this reason, we define $w'(a_2b_2) =(\wcp(a_2b_2),\wacp(a_2b_2))$  where 
\begin{align*}
\wcp(a_2b_2) & = \wc(C_2)  + \wc(A_2,B_2)+ \wc(A_2,C_2)+\wc(B_2,C_2)  \ \text{ and}\\
\wacp(a_2b_2) & =  \wac(C_2) +\wac(A_2,B_2)+\wac(A_2,C_2)+\wac(B_2,C_2) + \wr(C_2) \ .
\end{align*}
 
  \item In case of an odd complement 2-join,  we proceed symmetrically and give the real weight $\wr(C_2)$ as an extra-complete weight to $a_2b_2$. We thus define  $w'(a_2b_2) =(\wcp(a_2b_2),\wacp(a_2b_2))$ where:
  \begin{align*}
\wcp(a_2b_2) & =\wc(C_2)  + \wc(A_2,B_2)+ \wc(A_2,C_2)+\wc(B_2,C_2) + \wr(C_2) \text{ and}\\
\wacp(a_2b_2) & =  \wac(C_2) +\wac(A_2,B_2)+\wac(A_2,C_2)+\wac(B_2,C_2) \ .
\end{align*}

\end{itemize}
  
In order to recover information about the original trigraph after several steps of contraction, we need to introduce the notion of model. Intuitively, imagine that a weighted trigraph $(T,w)$ is obtained from an initial weighted trigraph $(T_0, w_0)$ by successive contractions, then we can partition the vertices of the original trigraph $T_0$ into subsets of vertices that have been contracted to the same vertex $v\in V(T)$ or the same switchable pair $uv\in \sigma(T)$. Moreover, the real weight of  a vertex $v\in V(T)$ is supposed to be the weight of the set of vertices that have been contracted to $v$. We also want  the strong adjacency and strong antiadjacency in $T$ to reflect the strong adjacency and strong antiadjacency in $T_0$. Finally, we want the extra-complete (resp. extra-anticomplete) weight in $T$ to stand for subsets of vertices of $T_0$ that have been deleted, but which were strongly complete (resp. strongly anticomplete) to (almost) all the rest of $T_0$.

Formally,
given a trigraph $T_0$ equipped with a virgin weight $w_0$, a weighted trigraph $(T, w)$ is a \defi{model} of $(T_0, w_0)$ if the following conditions are fulfilled (see Figure \ref{fig: partition mapping} for an example): 

\begin{figure}[t]
\center
\subfigure[A weighted trigraph $(T,w)$.]{\includegraphics[scale=0.9, page=3]{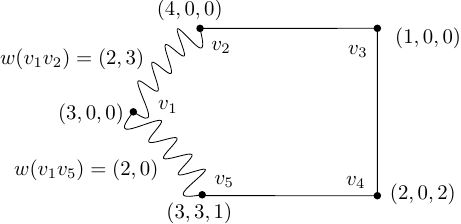}}\\
\subfigure[A weighted trigraph $(T_0, w_0)$ with virgin weight $w_0$ defined by $w_0(v)=(1,0,0)$ for every vertex $v\in V(T_0)$ and $w_0(uv)=(0,0)$ for every $uv\in \sigma(T_0)$.]{\includegraphics[scale=0.95, page=2]{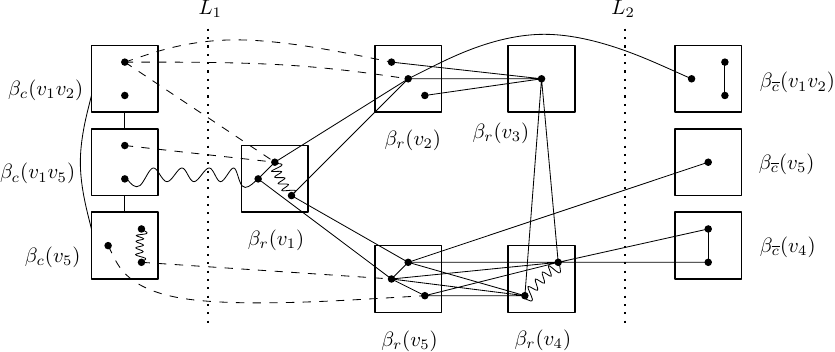}}
\caption{Illustration for the definition of a model: the weighted trigraph $(T,w)$ depicted in (a) is a model of the weighted trigraph $(T_0, w_0)$ depicted in (b), as witnessed by the partition map $\beta$ (empty teams are not depicted). Each vertex at the left-hand-side of dotted line $L_1$ is assumed to be strongly adjacent to every other vertex except if a non-edge is explicitly drawn (with a dashed edge for strong antiedge and with a wiggly edge for a switchable pair). Similarly, each vertex at the right-hand-side of dotted line $L_2$ is assumed to be strongly antiadjacent to every other vertex except if an edge is explicitly drawn.}
\label{fig: partition mapping}
\end{figure}
  
  \begin{itemize}
  \item \emph{The partition condition:} there exists a \defi{partition map} $\beta$ which,  to every vertex $v\in V(T)$ (resp. switchable pair $uv\in \sigma(T)$), assigns a triple $\beta(v)=(\br(v), \bc(v), \bac(v))$ (resp. a pair $\beta(uv)=(\bc(uv), \bac(uv))$)
  	 of (possibly empty) disjoint subsets of vertices of $T_0$.
  	 We define the \defi{team} of $v\in V(T)$ (resp. $uv\in \sigma(T)$) as $\bt(v)=\br(v)\cup \bc(v)\cup \bac(v)$ (resp. $\bt(uv)=\bc(uv)\cup \bac(uv)$). For convenience, $\br(v)$ is called the \defi{real team} of $v$ and for $x\in V(T)$ or $x\in \sigma(T)$,  $\bc(x)$ (resp. $\bac(x)$) is called the \defi{extra-complete team} (resp. \defi{extra-anticomplete team}) of $x$.
Moreover, any two teams must be disjoint and the union of all teams is $V(T_0)$. In other words, $V(T_0)$ is partitioned into teams, each team being itself divided into two or three disjoint parts.
Similarly to the weight function, for a subset of vertices $U\subseteq V(T)$ we define  $\beta(U) =(\br(U),\bc(U),\bac(U))$ where $\br(U)$ is the union of $\br(v)$ over all $v\in U$, and $\bc(U)$ (resp. $\bac(U)$) is the union of $\bc(x)$ (resp. $\bac(x)$) over all $x\in U$ and all $x=uv\in \sigma(T)$, where $u,v\in U$. Moreover, for two disjoint subsets of vertices $A,B\subseteq V(T)$,
$\beta(A,B) = (\bc(A,B), \bac(A,B)) $ where $\bc(A,B)$ (resp. $\bac(A,B)$) is the union of $\bc(ab)$ (resp. $\bac(ab)$) over all $a\in A$, $b\in B$ such that $ab\in \sigma(T)$.

  \item \emph{The weight condition:} 
 the total weight is preserved, \emph{i.e.} $\wt(T)=w_0(T_0)$ and for every  $v\in V(T)$, its real (resp. extra-complete, extra-anticomplete) weight  stands for the original weight of its real (resp. extra-complete, extra-anticomplete) team, and similarly for the switchable pairs. Formally,
for $v\in V(T)$, $\wr(v)=w_0(\br(v))$ and for $x\in V(T)$  or $ x\in \sigma(T)$, 
$\wc(x)= w_0(\bc(x))$ and $\wac(x)= w_0(\bac(x))$.
 
  \item \emph{The strong adjacency condition:} if $u,v\in V(T)$ are strongly adjacent (resp. strongly antiadjacent) in $T$, then $\br(u)$ and $\br(v)$ are strongly complete (resp. strongly anticomplete) in $T_0$. 
  \item \emph{The extra-condition:} Informally, the extra-complete team of $v$ (resp. $uv$) is strongly complete to every other extra-complete team, and is also strongly complete to every real team, except maybe the real team of $v$ (resp. of $u$ and $v$). The symmetric holds for the extra-anticomplete teams.
  Formally:  for every vertex $v$, $\bc(v)$ is strongly complete to every $\bc(x)$  for $x\in V(T)$, $x\neq v $ or $x\in \sigma(T)$, and to every $\br(y)$ for $y\neq v$.  For every switchable pair $uv\in \sigma(T)$, $\bc(uv)$ is strongly complete to every $\bc(x)$  for $x\in V(T)$ or $x\in \sigma(T)$, $x\neq uv$, and to every $\br(y)$ for $y\neq u, v$. For every vertex $v$, $\bac(v)$ is strongly anticomplete to every $\bac(x)$  for $x\in V(T)$, $x\neq v$ or $x\in \sigma(T)$, and to every $\br(y)$ for $y\neq v$.  For every switchable pair $uv\in \sigma(T)$, $\bac(uv)$ is strongly anticomplete to every $\bac(x)$  for $x\in V(T)$ or $x\in \sigma(T)$, $x\neq uv$, and to every $\br(y)$ for $y\neq u, v$.
  \end{itemize}

We are now ready for the proof, let us first provide a sketch: start from a trigraph $T_0$ with a balanced virgin weight $w_0$, in which we want to find a biclique or complement biclique of large weight. Iteratively contract it, and prove that at each step, the contraction is still a model of $(T_0,w_0)$ with a balanced weight. Stop either when the teams provide a  biclique or complement biclique of large weight in $(T_0,w_0)$, or when we reach a basic trigraph $(T,w)$. In the latter case, delete the extra-complete and extra-anticomplete weight, find a biclique or complement biclique of large weight in $(T,w)$, and convert it into a  biclique or complement biclique  of large weight in $T_0$.

\begin{lemma}
\label{lem: contraction preserve model}
Let $(T_0,w_0)$ be a weighted trigraph such that  $T_0\in \mathcal{F}$ has no balanced skew-partition and $w_0$ is a balanced virgin weight. Let $(T,w)$ be a model of $(T_0,w_0)$ such that $w$ is balanced and $T$ is a non-basic trigraph of $\mathcal{F}$ with no balanced skew-partition. Then at least one of the following holds:
\begin{itemize}
\item There exists a biclique or complement biclique in $(T_0,w_0)$ of weight at least $ \D  w_0(T_0)$.
\item Any contraction $(T',w')$ of $(T,w)$ is a model of $(T_0,w_0)$ and $w'$ is balanced. Moreover, $T'$ is a trigraph of $\mathcal{F}$ with no balanced skew-partition.
\end{itemize}
\end{lemma}

\begin{figure}[t]
\center
\subfigure[Case 1: even 2-join. ]{\label{fig: contraction preserve model even}\includegraphics[scale=0.9, page=1]{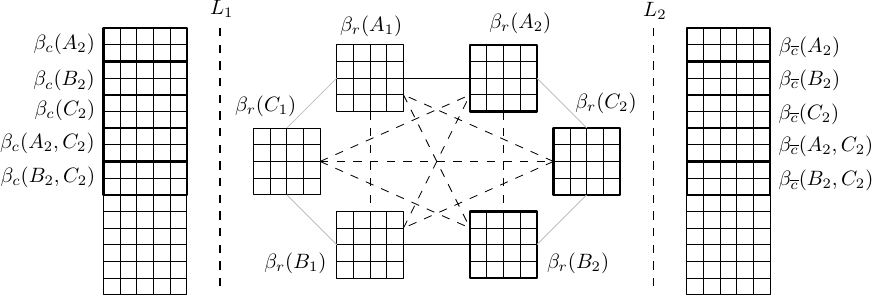}}
\subfigure[Case 2: odd 2-join]{\label{fig: contraction preserve model odd}\includegraphics[scale=0.9, page=2]{fig/Model2Join}}
\caption{Illustration for the proof of Lemma \ref{lem: contraction preserve model}. 
The little boxes show how $\beta$ partitions $V(T_0)$, certifying that $(T,w)$ is a model of $(T_0, w_0)$. Boxes with bold font show groups of teams that are merged together by $\beta'$, certifying that the contraction $(T',w')$ of $(T,w)$ is still a model of $(T_0,w_0)$. For case 2, red boxes highlight the most tricky part of the proof concerning the extra-complete and extra-anticomplete teams of the new switchable pair $a_2b_2$.
The extra-complete teams are depicted on the left of dotted line $L_1$, extra-anticomplete teams are depicted on the right of dotted line $L_2$, and real teams are in between. For a better drawing, adjacencies assumed for the extra-condition are implied but not depicted.
Grey lines indicate that there may or may not be some edges. Dashed lines link strongly anticomplete teams, and straight lines link strongly complete teams.}
\label{fig: contraction preserve model}
\end{figure}

 \begin{proof}
First of all, let us check that the second item is well-defined: by assumption, $T$ is a trigraph of $\mathcal{F}$ with no balanced skew-partition and is not basic, thus $T$ has a 2-join or a complement 2-join $(X_1,X_2)$ with split $(A_1, B_1, C_1,
A_2, B_2, C_2)$. Consequently, the contraction of $(T,w)$ with respect to this split is well-defined and $T'$ is the block of decomposition $T_{X_1}$ or $T_{X_2}$. By Theorem \ref{l:stayBerge}, $T'$ is a trigraph of $\mathcal{F}$ with no balanced skew-partition.
We assume that $\wt(X_1)\geq \wt(T)/2$ (consequently $T'=T_{X_1}$) and that $(X_1,X_2)$ is a 2-join (otherwise we exchange $T$ and $\comp{T}$).

\textbf{Case 1: $(X_1, X_2)$ is an even 2-join.}

We first prove that $(T',w')$ is a model of $(T_0,w_0)$. Since $(T,w)$ is a model of $(T_0,w_0)$, there exists a partition map $\beta$ that certifies it. Let us build $\beta'$ a partition map for $(T',w')$ in the following natural way (see Figure \ref{fig: contraction preserve model even}):
\begin{itemize}
\item For every $v\in X_1$, $\beta'(v)=\beta(v)$ and for every $u,v\in X_1$, $uv\in \sigma(T')$, $\beta'(uv)=\beta(uv)$.
\item $\beta'(a_2)=\beta(A_2)$, $\beta'(b_2)=\beta(B_2)$, $\beta'(c_2)=\beta(C_2)$.
\item $\beta'(a_2c_2)=\beta(A_2, C_2)$ and $\beta'(b_2c_2)=\beta(B_2,C_2)$.
\end{itemize}

It is quite easy to check  that the weight condition, the strong adjacency condition and the extra-condition hold. Let us explain here only some parts in detail: concerning the strong adjacency condition, observe that the strong adjacency or strong antiadjacency between $a_2,b_2$ and $c_2$ on one hand, and any $v_1\in X_1$ on the other hand, mimic the behavior of the 2-join, by definition of the block of decomposition. Moreover, since the 2-join is even, there is no edge between $A_2$ and $B_2$ in $T$, which explains the strong antiedge between $a_2$ and $b_2$.
As for the extra-condition, we can observe that  the new extra-complete (resp. extra-anticomplete) teams are obtained by merging former extra-complete (resp. extra-anticomplete) teams. Let us study an example: let $v\in \bcp(a_2c_2)$, then by definition there exists $ac\in \sigma(T)$ such that $a\in A_2, c\in C_2$ and $v\in \bc(ac)$. Since $(T,w)$ is a model of $(T_0,w_0)$, $v$ is strongly complete to every other extra-complete teams of $\beta$ except the one it belongs to (and thus to every extra-complete teams of $\beta'$ except $\bcp(a_2c_2)$), and $v$ is also strongly complete to every real team except maybe $\br(a)$ and $\br(c)$. But $a\in A_2$ and $c\in C_2$ so $\br(a)\subseteq \brp(a_2)$ and $\br(c)\subseteq \brp(c_2)$. Consequently, $v$ is strongly complete to every real teams, except maybe $\brp(a_2)$ and $\brp(c_2)$: this is what we require for a member of $\bcp(a_2c_2)$.

We now have to see if $w'$ is balanced. First of all, 
$$\wcp(T')+\wacp(T')=\wc(T)+\wac(T) \leq \D[7] \wt(T)=\D[7]\wtp(T')\ .$$
Moreover, observe that $\wt(X_1)=\wr(A_1)+\wr(B_1)+\wr(C_1)+\wc(X_1)+\wac(X_1)\ .$ 

\noindent But $\wt(X_1)\geq \wt(T)/2$ and $\wc(X_1)+\wac(X_1)\leq  \D[7] \wt(T)$ so 
$$\max(\wr(A_1),\wr(B_1),\wr(C_1))\geq \frac{1}{3}\left(\frac{1}{2} -\frac{7}{55} \right) \wt(T)\geq \D \wt(T)\ .$$ Since the other cases are handled similarly, we assume that $\wr(A_1)\geq \D \wt(T)$.
Each of $\brp(a_2)$, $\brp(b_2)$ and $\brp(c_2)$ is either strongly complete or strongly anticomplete to $\br(A_1)$ whose weight is $w_0(\br(A_1))=\wr(A_1)\geq \D \wt(T)$, so if $\max(\wrp(a_2), \wrp(b_2), \wrp(c_2))\geq \D \wt(T)$, we find a biclique or a complement biclique of large enough weight in $T_0$, and the first item holds.
Otherwise, observe that every extra-complete team among $\bcp(a_2)$, $\bcp(b_2)$, $\bcp(c_2)$, $\bcp(a_2c_2)$, $\bcp(b_2c_2)$ is strongly complete to all the real teams $\brp(x)$ for  $x\in X_1$, thus if one of them has weight $\geq \D \wt(T)$ in $(T_0,w_0)$, we find a biclique in $T_0$ of large enough weight, and the first item holds. Thus $\wcp(x)\leq \D \wt(T)$ for every $x\in \set{a_2, b_2, c_2}$ and $x\in \set{a_2c_2, b_2c_2}$.
 By similar arguments,  we also have 
 $\wacp(a_2)$, $\wacp(b_2)$, $\wacp(c_2)$, $\wacp(a_2c_2)$, $\wacp(b_2c_2) \leq \D \wt(T)$
 otherwise we find a large complement biclique in $T_0$ and conclude with the first item. 
 Hence $w'$ is balanced and we conclude with the second item.

\textbf{Case 2: $(X_1, X_2)$ is an odd 2-join.}

As in the previous case, we begin with proving that $(T',w')$ is a model of $(T_0,w_0)$. Since $(T,w)$ is a model of $(T_0,w_0)$, there exists a partition map $\beta$ that certifies it. Let us build $\beta'$ a partition map for $(T',w')$ in the following  way (see Figure \ref{fig: contraction preserve model odd}):
\begin{itemize}
\item For every $v\in X_1$, $\beta'(v)=\beta(v)$ and for every $u,v\in X_1$, $uv\in \sigma(T')$, $\beta'(uv)=\beta(uv)$.
\item $\beta'(a_2)=\beta(A_2)$, $\beta'(b_2)=\beta(B_2)$.
\item For the switchable pair $a_2b_2$, we follow the same approach as before for the weight function because we do not want to loose track from the teams of type $\bt(c)$ for $c\in C_2$ or $\bt(vc)$ for $c\in C_2$, $vc\in \sigma(T)$.   Formally, we define $\beta'(a_2b_2) =(\bcp(a_2b_2),\bacp(a_2b_2))$ where:
\begin{align*}
\bcp(a_2b_2)&= \bc(A_2,B_2)\cup \bc(A_2,C_2)\cup \bc(B_2,C_2)\cup \bc(C_2)\\
\bacp(a_2b_2)& = \bac(A_2,B_2)\cup \bac(A_2,C_2)\cup \bac(B_2,C_2)\cup \bac(C_2) \cup \br(C_2) \ .
\end{align*}
\end{itemize}

Once again, we easily see that the weight condition and the strong adjacency condition are ensured with the same arguments as in Case 1. As for the extra-condition, the only interesting case is concerned with $\bacp(a_2b_2)$: let $v\in \bacp(a_2b_2)\subseteq V(T_0)$. The goal is to prove that $v$ is strongly anticomplete to every other extra-anticomplete team of $\beta'$, and to every real team of $\beta'$ except maybe the real team of $a_2$ and the real team of $b_2$. 
By definition, $v$ must belong to one of the five subsets constituting $\bacp(a_2b_2)$. 
If $v\in \bac(A_2,C_2)$, then there exists $ac\in \sigma(T)$ such that $a\in A_2$, $c\in C_2$ and $v\in \bac(ac)$. Since $(T,w)$ is a model of $(T_0,w_0)$, $v$ is strongly anticomplete to every other extra-anticomplete team of $\beta$, thus of $\beta'$, and $v$ is also strongly anticomplete to every real teams except maybe $\br(a)$ and $\br(c)$. But $\br(a)\subseteq\brp(a_2)$ and $\br(c)\subseteq \br(C_2)\subseteq \bacp(a_2b_2)$ so $v$ is strongly anticomplete to every real team except maybe $\brp(a_2)$.
The cases $v\in \bac(B_2,C_2)$,  $v\in \bac(A_2,B_2)$ and $v\in \bac(C_2)$ are handled with similar arguments.
Finally if $v\in \br(C_2)$, then there exists $c\in C_2$ such that $v\in \br(c)$. By definition of a 2-join, $c$ is strongly anticomplete to $X_1$ in $T$, so since $(T,w)$ is a model of $(T_0,w_0)$,  $v$ is strongly anticomplete to every real team $\br(x)$ with $x\in X_1$, \emph{i.e.} to every real team of $\beta'$ except maybe $\brp(a_2)$ and $\brp(b_2)$.
Moreover, by the extra-condition on $(T,w)$, $\br(C_2)$ is strongly anticomplete to every extra-anticomplete team of $\beta$ except those included in $\bac(C_2)$, $\bac(A_2,C_2)$ or $\bac(B_2,C_2)$. But those three are all included in $\bacp(a_2b_2)$, so $v$ is strongly anticomplete to every extra-anticomplete team
 different from $\bacp(a_2b_2)$.

 Let us now check that $w'$ is balanced. With the same argument as in Case 1, we obtain that   $\max(\wr(A_1),\wr(B_1),\wr(C_1))\geq \D \wt(T)$ so $\max(\wr(a_2), \wr(b_2), \wr(C_2))\leq \D \wt(T)$, and similarly $\max(\wcp(x), \wacp(x))\leq \D \wt(T)$ for $x=a_2, b_2$ or $a_2b_2$.

 Finally, 
 we want to prove that $\wcp(T')+\wacp(T')\leq \D[7] \wtp(T')$.
 Assume not, then since
 $$\wcp(T')+\wacp(T')=\wc(T)+\wac(T)+\wr(C_2)  \quad \text{ and } \quad \wr(C_2)\leq \D \wt(T) \ ,$$
 we have $\wc(T)+\wac(T)\geq \D[6] \wt(T)$.
 Thus one of $\wc(T)$ or $\wac(T)$, say $\wc(T)$, is at least $\D[3] \wt(T)$. Since every extra-complete team $\bc(x)$ for $x\in V(T)$ or $x\in \sigma(T)$ has weight at most $ \D\wt(T)$, we can split $\bc(T)$ into two parts $(X,Y)$ such that no extra-complete team intersects both $X$ and $Y$, and such that both $w_0(X)$ and $ w_0(Y)$ are at least $ \D \wt(T)$. Since each extra-complete teams is strongly complete to every other extra-complete team, $(X,Y)$ is a biclique, and its weight is at least $\D w_0(T_0)$: the first item of the lemma holds.
 \end{proof}

Before going to the case of basic trigraphs, we need a technical lemma that will be useful to handle the line trigraph case. A \emph{multigraph} $G=(V,E)$ is a generalization of a graph where
 $E$ is a multiset of pairs of distinct vertices: there can be several edges between two distinct vertices. The number of edges is the cardinality of the multiset $E$.
  An edge $uv$ has two \defi{endpoints} $u$ and $v$. The \defi{degree} of  $v\in V(G)$ is $d(v)=|\{e\in E|v \text{ is an endpoint of } e\}|$.

\begin{lemma}
\label{partition edges bipartite}
Let $G$ be a bipartite multigraph $(A,B)$ with $m$ edges and with maximum degree less than $m/3$. Then there exist $E_1, E_2\subseteq E$ such that $|E_1|, |E_2|\geq m/48$ and if $e_1\in E_1, e_2 \in E_2$ then $e_1$ and $e_2$ do not have a common extremity.
\end{lemma}

\begin{proof}
The \emph{score} $S(U,U')$ of a bipartition $(U,U')$ of $V(G)$ is defined as the number of unordered pairs of edges $\{ uv, u'v'\}\subseteq E$ such that $u,v\in U$ and $u',v'\in U'$ (\emph{i.e.} $uv$ is on one side of the partition and $u'v'$ is on the other side), that is to say $S(U,U')=|E\cap U^2| \cdot |E\cap U'^2|$. Let $\gamma$ be the number of unordered pairs of edges $e_1, e_2\in E$ such that $e_1$ and $e_2$ have no common endpoint. The expectation of $S(U,U')$ when $(U,U')$ is a random uniform partition of $V(G)$ is $\gamma/8$, so there exists a partition $(U,U')$ such that $S(U,U')\geq \gamma/8$. Assume now that $\gamma\geq m^2/6$ and let 
$E_1=\set{uv\in E \ | \ u,v\in U }$ and $E_2=\set{u'v'\in E \ | \ u',v'\in U' }$.
 Then $|E_1|, |E_2|\geq m/48$, otherwise
$$S(U,U')=|E_1|\cdot|E_2|<\frac{m^2}{48}=\frac{\gamma}{8} \ , $$ a contradiction. So $E_1$ and $E_2$ satisfy the requirements of the lemma. We finally have to prove that $\gamma\geq m^2/6$. For a given $e_1=uv\in E$, the number of edges different from $e_1$ which have a common endpoint with $e_1$ it at most $(d(u)+d(v)-2)\leq 2m/3 -2$. Consequently,
$$\gamma \geq \frac{1}{2}\sum_{e_1\in E} \left\vert \{e_2 \in E \  | \ e_1\cap e_2=\emptyset \}\right\vert \geq
\frac{1}{2} \cdot m \cdot \left( ( m-1 )-\left(\frac{2m}{3} -2\right) \right) \geq \frac{m^2}{6} \ .$$
\end{proof}

We can now give the proof for the case of basic trigraphs.

\begin{lemma}
\label{lem: basic SEH}
Let $(T,w)$ be a weighted trigraph such that $T$ is a basic trigraph and $w$ is balanced.
Then $T$ admits a biclique or a complement biclique $(X,Y)$ of weight $\min(\wr(X),\wr(Y))\geq \D \wt(T)$.
\end{lemma}

\begin{proof}[Proof of Lemma \ref{lem: basic SEH}]
Let us transform the weight $w$ into a virgin weight $w_0$ defined as $w_0(v)=(\wr(v),0,0)$ for every vertex $v$ and $w_0(uv)=(0,0)$ for every $uv\in \sigma(T)$. In other words, all the non-real weight is deleted. Since $w$ is balanced, $\wc(T)+\wac(T)\leq \D[7] \wt(T)$ so 
$$w_0(T)=\wt(T)-(\wc(T)+\wac(T))\geq \left(1-\frac{7}{55}\right) \wt(T) \ .$$ Now it is enough to find a biclique or a complement biclique in $(T,w_0)$ with weight $\geq \frac{1}{48} \cdot w_0(T)$ since $ \frac{1}{48} \cdot w_0(T) \geq \D \wt(T)$. Observe that every vertex still has weight at most $\D \wt(T)\leq \frac{1}{48} \cdot w_0(T)$. Since the property is self-complementary, we have only three cases to examine.

If $T$ is a bipartite trigraph, then $V(T)$ can be partitioned into two strong stable sets. One of them has weight at least $\geq \frac{w_0(T)}{2}\geq \frac{1}{16} \cdot w_0(T)$. Moreover, each vertex has weight at most $\frac{1}{48} \cdot w_0(T)$ so we can split the stable set into two parts, each of weight $\geq \frac{1}{48} \cdot w_0(T)$.

If $T$ is a doubled trigraph, then observe that $V(T)$ can be partitioned into two strong stable sets (the first side of the good partition) and two strong cliques (the second side of the good partition). Hence, one of these strong stable sets or cliques has weight $\geq w_0(T)/4$, and, by the same argument as above, we can split it in order to obtain a biclique or a complement biclique of weight $\geq \frac{1}{48} \cdot  w_0(T)$.

It is slightly more complicated if $T$ is a line trigraph. If there exists a clique $K$ of weight
at least $\frac{1}{16} \cdot w_0(T)$, then it is a strong clique: indeed, by definition of a line trigraph, every clique of size at least three is a strong clique; moreover, a clique of size at most two has weight at most $\frac{1}{24} \cdot w_0(T)$. 
Then we can split $K$ as above and get a biclique of weight $\frac{1}{48} \cdot w_0(T)$.
Assume now that such a clique does not exist and 
let $F$ be the full realization of $T$ (the graph obtained from $T$ by replacing every switchable pair by an edge). Observe that a complement biclique in $F$ is also a complement biclique in $T$. By definition of a line trigraph, $F$ is the line graph of a bipartite graph $G$. Instead of keeping positive integer weight on the edges of $G$, we convert $G$ into a multigraph $G'$ by transforming  each edge $uv$ of weight $s$ into $s$ edges $uv$. The inequality $w_0(K)\leq 1/16 \cdot w_0(T)$ for every clique $K$ of $T$ implies that the maximum degree of a vertex in $G'$ is at most $1/16 \cdot w_0(T)$. Lemma \ref{partition edges bipartite} proves the existence of two subsets $E_1, E_2$ of edges of $G'$ such that $|E_1|, |E_2|\geq w(V)/48$ and if $e_1\in E_1, e_2 \in E_2$ then $e_1$ and $e_2$ do not have a common extremity. This corresponds to a complement biclique in $F$ and thus in $T$ of weight $\geq \frac{1}{48} \cdot w_0(T)$.
\end{proof}

We can now prove the main theorem of this section:

\begin{theorem}
\label{th: SEH trigraph sans BSP}
Let $T_0$ be a trigraph of $\mathcal{F}$ with no balanced skew-partition, equipped with a virgin balanced weight $w_0$. Then $T_0$ admits a biclique or a complement biclique of weight at least $\D w_0(T_0)$.

\end{theorem}

In particular, this proves that the class of Berge graphs with no balanced skew-partition has the Strong \EH \ property, as announced in Theorem \ref{th: SEH berge no BSP} :


\begin{proof}[Proof of Theorem \ref{th: SEH berge no BSP}]
Let $T$ be a trigraph of $\mathcal{F}$ with no balanced skew-partition and $w_0$ be the virgin weight defined by $w_0(v)=(1,0,0)$ for every vertex $v\in V(T)$. Assume that $|V(T)|\geq 3$. The goal is to prove that $T$ admits a biclique or a complement biclique of size at least $|V(T)|/55$.
If $|V(T)|\geq 55$, then $w_0$ is balanced and $w_0(T)=|V(T)|$, so we apply Theorem \ref{th: SEH trigraph sans BSP}. Otherwise, since $|V(T)|\geq3$ and $T\in \mathcal{F}$, $T$ contains at least one strong edge or one strong antiedge: this gives a biclique or a complement biclique of size $1\geq \D |V(T)|$. 
\end{proof}

\begin{proof}[Proof of Theorem \ref{th: SEH trigraph sans BSP}]
Start with $(T,w)=(T_0,w_0)$ 
and iteratively contract $(T,w)$ with the help of Lemma \ref{lem: contraction preserve model} until either item (i) of the lemma occurs, which concludes the proof, or we get a basic trigraph $(T,w)$ which is a model of $(T_0, w_0)$ and where $w$ is balanced.
%
%
%
%
In the latter case, by Lemma \ref{lem: basic SEH}, $T$ admits a biclique or a complement biclique, say a biclique, of weight $\geq\D \wt(T)$. This means that there exists a pair $(X,Y)$ of disjoint subsets of vertices of $T$ such that $\wr(X), \wr(Y)\geq \D \wt(T)$ and $X$ is strongly complete to $Y$. Since $(T,w)$ is a model of $(T_0,w_0)$, we transform $(X,Y)$ into a biclique of large weight in $T_0$ as follows: let $\beta$ be the partition map for $(T,w)$ and let $X'=\br(X)\subseteq V(T_0)$ and $Y'=\br(Y)\subseteq V(T_0)$. 
By the \emph{strong adjacency condition} in the definition of a model, $X'$ is strongly complete to $Y'$ in $T_0$ since $X$ is strongly complete to $Y$ in $T$.
Moreover, by the \emph{weight condition}, we have $w_0(\br(X))=\wr(X)$ and $w_0(\br(Y))=\wr(Y)$. But then $w_0(X'), w_0(Y') \geq \D \wt(T)=\D w_0(T_0)$, which concludes the proof.
\end{proof}

\subsection{In the closure  $\overline{\mathcal{C}^{\leq k}}$ of $\mathcal{C}$ by generalized $k$-join}

In fact, the method of contraction of a $2$-join used in the previous subsection can easily be adapted to a generalized $k$-join. We only require that the basic class $\mathcal{C}$ of graphs is hereditary and has the Strong \EH \ property. We invite the reader to refer to Subsection \ref{sec: k-join} for the definitions of a generalized $k$-join and the classes $\mathcal{C}^{\leq k}$ and $\overline{\mathcal{C}^{\leq k}}$. The proof is even much easier than for Berge trigraphs with no balanced skew-partition because there is no problematic case such as the odd 2-join, where no vertex keeps track of the deleted  part $C_2$. Consequently, there is no need to introduce extra-complete and extra-anticomplete weight, and from now on, we simply work with non-negative integer weight on the vertices. A biclique  (resp. complement biclique) in $T$ is still a pair $(X,Y)$ of subsets of vertices such that $X$ is strongly complete (resp. strongly anticomplete) to $Y$. Its weight is defined as $\min(w(X), w(Y))$.

We now define the contraction of a weighted trigraph $(T,w)$ containing a generalized $k$-join of $T_1$ and $T_2$. We follow the notation introduced in the definition of the generalized $k$-join, in particular $V(T)$ is partitioned into $(A_1, \ldots , A_r, B_1, \ldots, B_s)$. Without loss of generality, assume that $w(\cup_{j=1}^{r}A_j)\geq w(\cup_{i=1}^{s}B_i)$. Then the \defi{contraction} of $T$ is the weighted trigraph $(T',w')$ with $T'=T_1$ and $w'$ defined by $w'(v)=w(v)$ if $v\in  \cup_{j=1}^{r}A_j$, and $w'(b_i)=w(B_i)$ for $1\leq i\leq s$. 

Finally, the definition of \defi{model} is also much simpler in this setting. Given a weighted trigraph $(T_0, w_0)$, we say that $(T,w)$ is a model of $(T_0, w_0)$ if the following conditions hold:

\begin{itemize}
  \item \emph{The partition condition:} there exists a \defi{partition map} $\beta$ which assigns to every vertex $v\in V(T)$ a subset $\beta(v)\subseteq V(T_0)$ of vertices of $T_0$, called the \defi{team} of $v$.
  	Moreover, any two teams are disjoint and the union of all teams is $V(T_0)$. Intuitively, the team of $v$ will contain all the vertices of $V(T_0)$ that have been contracted to $v$.
  	Similarly as before, for a subset $U\subseteq V(T)$ of vertices, we define
$\beta(U)$ to be the union of $\beta(u)$ over all $u\in U$.
  \item \emph{The weight condition:} 
 $w(T)=w_0(T_0)$ and for all  $v\in V(T)$, $w(v)=w_0(\beta(v))$.
  \item \emph{The strong adjacency condition:} if two vertices $u$ and $v$ are strongly adjacent (resp. strongly antiadjacent)  in $T$, then $\beta(u)$ and $\beta(v)$ are strongly complete (resp. strongly anticomplete) in $T_0$. 
  \end{itemize}
  
Here are two last definitions before giving the proof.  Given $0<c<1/2$ and a trigraph $T$, a weight function $w:V(T) \mapsto \mathbb{N}$ is \defi{$c$-balanced} if for every vertex $v\in V(T)$, $w(v)\leq c \cdot w(T)$. A hereditary class  $\mathcal{C}$ of graphs  is said \defi{$c$-good} if the following holds: 
  for every $G\in \mathcal{C}$ with at least 2 vertices and for every $c$-balanced weight function $w$ on $V(G)$, $G$ admits a biclique or a complement biclique of weight $\geq c \cdot w(G)$.
We are now ready to obtain the following result and its corollary:

\begin{theorem}\label{th: patates k-join}
  Let $k\geq 1$, $0<c<1/2$ and assume that $\mathcal{C}$ is a $ck$-good class of graphs.
Then for every $T_0 \in \overline{\mathcal{C}^{\leq k}}$ containing at least one strong edge or one strong antiedge, and for every $c$-balanced weight function $w_0$,
the weighted trigraph $(T_0,w_0)$
has a biclique or a complement biclique  of weight at least $c \cdot w_0(T_0)$.
\end{theorem}

\begin{coro} 
 Let $k\geq 1$, $0<c<1/2$ and $\mathcal{C}$ be a $ck$-good class of graphs. Let $(T_0,w_0)$ be a weighted trigraph such that $w_0(v)=1$ for every $v\in V(T_0)$ and $T_0 \in \overline{\mathcal{C}^{\leq k}}$. Then $(T_0, w_0)$ admits a biclique or a complement biclique of size $c \cdot |V(T_0)|$, provided that $T_0$ has at least one strong edge or one strong antiedge.
\end{coro}

\begin{proof}
If $V(T_0)<1/c$, then one strong edge or one strong antiedge suffices to form a biclique or a complement biclique of size $1\geq c \cdot |V(T_0)|$. Otherwise, $w_0$ is $c$-balanced so we apply Theorem \ref{th: patates k-join}.
\end{proof}

To begin with, we need a counterpart of Lemma \ref{lem: contraction preserve model} to prove that the contraction of a model is still a model:

\begin{lemma}
\label{lem: contraction k-join}
Let $\mathcal{C}$ be a class of graphs, $k\geq 1$, and $0<c<1/2k$.
Let $(T_0,w_0)$ be a weighted trigraph such that $T_0\in \overline{\mathcal{C}^{\leq k}}$ and $w_0$ is $c$-balanced. Then if $(T,w)$ is a model of $(T_0,w_0)$ with $T\in \overline{\mathcal{C}^{\leq k}}$ but $T \notin \mathcal{C}^{\leq k}$ and if $w$ is $c$-balanced, at least one of the following holds:

\begin{itemize}
\item There exists a biclique or a complement biclique in $T_0$ of weight $\geq c\cdot w_0(T_0)$.
\item The contraction $(T',w')$ of $(T,w)$ is also a model of $(T_0, w_0)$; moreover  $T'\in \overline{\mathcal{C}^{\leq k}}$ and $w'$ is $c$-balanced.
\end{itemize}

\end{lemma}

\begin{proof} Since $T\notin \mathcal{C}^{\leq k}$,  $T$ is the generalized $k$-join between two trigraphs $T_1$ and $T_2$. Following the same notation as in the definition of a $k$-join, we assume that $V(T_1)$ is partitioned into $(A_1, \ldots , A_r, \set{b_1, \ldots, b_s})$ and $V(T_2)$ into $(\set{a_1, \ldots,  a_r}, B_1, \ldots, B_s)$ with $r,s\leq k$.
Without loss of generality, we can assume that $w(\cup_{j=1}^{r}A_j)\geq w(\cup_{i=1}^{s}B_i)$. 
Since $r\leq k$,  there exists $j_0$ such that \mbox{$w(A_{j_0})\geq \frac{1}{2k} \cdot w(T)$}. Now if there exists $i_0$ such that $w(B_{i_0})\geq c \cdot w(T)$, 
then $(A_{j_0}, B_{i_0})$ is a biclique or a complement biclique, by definition of a generalized $k$-join, and its weight is $\geq c \cdot w(T)=c \cdot w_0(T_0)$, thus item (i) holds. 
Otherwise, the goal is to prove that the contraction $(T', w')$ of $(T,w)$ is also a model of $(T_0, w_0)$, where $T'=T_1\in \overline{\mathcal{C}^{\leq k}}$ and $w'$ defined as above by $w'(v)=w(v)$ if $v\in  \cup_{j=1}^{r}A_j$, and $w'(b_i)=w(B_i)$ for $1\leq i\leq s$. 
Observe that $w'(T')=w(T)$ and that $w'$ is $c$-balanced. 
Moreover, let $\beta$ be the partition map certifying that  $(T,w)$ is a model of $(T_0, w_0)$.
 We can easily see that $(T', w')$ is a model of $(T_0,w_0)$ by defining $\beta'(v)=\beta(v)$ if $v\in  \cup_{j=1}^{r}A_j$, and $\beta'(b_i)=\beta(B_i)$ for every $1\leq i \leq s$. We can check that all the conditions are ensured. This concludes the proof.
\end{proof}

For the \emph{basic} case, we need to adapt our assumption on $\mathcal{C}$ to make it work on $\mathcal{C}^{\leq k}$:

\begin{lemma}
\label{lem: SEH basic k-join}
Let $k\geq 1$, $0<c<1/2$ and $\mathcal{C}$ be a $ck$-good class of graphs.
Let $(T,w)$ be a weighted trigraph such that $T\in \mathcal{C}^{\leq k}$, $w$ is $c$-balanced and $T$ contains at least one strong edge or one strong antiedge. Then $T$ admits a biclique or a complement biclique of weight $c \cdot w(T)$.
\end{lemma}

\begin{proof} 
For every switchable component of $T$, select the vertex with the largest weight and delete the others. We obtain a graph $G\in \mathcal{C}$ and define $w_G(v)=w(v)$ on its vertices. Observe that $w_G(G)\geq w(T)/k$ since every switchable component has size $\leq k$, and that  $w_G(v)=w(v)\leq c \cdot w(T)\leq ck \cdot w_G(G)$ for every $v\in V(G)$. 
Moreover, $G$ has at least 2 vertices since $T$ has at least two different switchable components.
Since $\mathcal{C}$ is $ck$-good, there exists a biclique or complement biclique$(V_1, V_2)$
 in $G$  such that $w_G(V_1), w_G(V_2)\geq ck \cdot w(G)$.
 Then $(V_1, V_2)$ is also a biclique or complement biclique in $T$ with the same weight 
 $\geq ck \cdot w_G(G)\geq c \cdot w(T)$ .
\end{proof}

\begin{proof}[Proof of Theorem \ref{th: patates k-join}]
Let $(T_0,w_0)$ be a weighted trigraph such that  $T_0 \in \overline{\mathcal{C}^{\leq k}}$ has at least one strong edge or one strong antiedge, and such that  $w_0$ is $c$-balanced. 
  Start with $(T,w)=(T_0,w_0)$ and keep contracting $(T,w)$ while
   $T \notin \mathcal{C}^{\leq k}$.
  By Lemma \ref{lem: contraction k-join}, 
at each step we know that  $T \in \overline{\mathcal{C}^{\leq k}}$, $(T,w)$ is a model of $(T_0, w_0)$ and $w$ is $c$-balanced, or we find  a biclique or a complement biclique of weight $c \cdot w_0(T_0)$ in $(T_0,w_0)$, in which case we can directly conclude.
 In the former case, we stop when $T\in \mathcal{C}^{\leq k}$.
By definition of a contraction, $T$ has at least one strong edge or one strong antiedge. Since $\mathcal{C}$ is $ck$-good, apply Lemma \ref{lem: SEH basic k-join} to get a biclique or complement biclique $(V_1, V_2)$ in $T$ of weight at least $c \cdot w(T)$.  Let $\beta$ be a partition map  certifying that $(T,w)$ is a model of $(T_0, w_0)$. Then $(\beta(V_1), \beta(V_2))$ is a biclique or complement biclique in $T_0$ according to the strong adjacency condition. 
We can now conclude by the weight condition:
$$\min(w_0(\beta(V_1)), w_0(\beta(V_2)))= \min(w(V_1), w(V_2))\geq c\cdot w(T)=c \cdot w_0(T_0) \ .$$
\end{proof}

\bibliographystyle{plain}
\bibliography{bibli_skew}

\end{document}